\documentclass[%
 reprint,
 amsmath,amssymb,
 aps,
]{revtex4-2}

\usepackage{amsthm}
\usepackage{graphicx}% Include figure files
\usepackage{subcaption}
\usepackage{dcolumn}% Align table columns on decimal point
\usepackage{bm}% bold math
\usepackage{xcolor}
\usepackage[bottom]{footmisc}
\usepackage{fancyhdr}

\newtheorem{definition}{Definition}
\newtheorem{example}{Example}
\newtheorem{theorem}{Theorem}
\newtheorem{lemma}{Lemma}
\newtheorem{proposition}{Proposition}
\newtheorem{remark}{Remark}
\newtheorem{corollary}{Corollary}

\begin{document}

\preprint{APS/123-QED}

\title{Characterizing High Schmidt Number Witnesses in Arbitrary Dimensions System}% Force line breaks with \\

\author{Liang Xiong}
 \altaffiliation{liang.xiong@polyu.edu.hk}%Lines break automatically or can be forced with \\
\author{Nung-sing Sze}%
 \email{raymond.sze@polyu.edu.hk (Corresponding author)}
\affiliation{Department of Applied Mathematics, The Hong Kong Polytechnic University, Hung Hom, Hong Kong, China
}%

%\author{Liang Xiong$^{1,2,}$}
%\email{liang.xiong@polyu.edu.hk}
%
%\author{Nung-sing Sze$^{2,}$}%
% \email{raymond.sze@polyu.edu.hk (Corresponding author)}
% \affiliation{$^{1}$School of Computer Science and Technology, Dongguan University of Technology, Dongguan, China\\
% $^{2}$Department of Applied Mathematics, The Hong Kong Polytechnic University, Hung Hom, Hong Kong, China}
% %\affiliation{Department of Applied Mathematics, The Hong Kong Polytechnic University, Hung Hom, Hong Kong, China}
% %Lines break automatically or can be forced with \\

%\affiliation{Department of Applied Mathematics, The Hong Kong Polytechnic University, Hung Hom, Hong Kong, China


%\affiliation{
%\affiliation{
%\affiliation{%

\date{\today}% It is always \today, today,

\begin{abstract}
A profound comprehension of quantum entanglement is crucial for the progression of quantum technologies. The degree of entanglement can be assessed by enumerating the entangled degrees of freedom, leading to the determination of a parameter known as the Schmidt number. In this paper, we develop an efficient analytical tool for characterizing high Schmidt number witnesses for bipartite quantum states in arbitrary dimensions.
Our methods not only offer viable mathematical methods for constructing high-dimensional Schmidt number witnesses in theory but also simplify the quantification of entanglement and dimensionality.
Most notably, we develop high-dimensional Schmidt number witnesses within arbitrary-dimensional systems, with our Schmidt witness coefficients relying solely on the operator Schmidt coefficient. Subsequently, we demonstrate our theoretical advancements and computational superiority by constructing Schmidt number witnesses in arbitrary dimensional bipartite quantum systems with Schmidt numbers four and five.
\begin{description}
\item[Keyword]
Schmidt number; Operator Schmidt decomposition; nonnegative matrices; entanglement witnesses
\end{description}
\end{abstract}

\maketitle
%\begin{keyword}
%Absolutely separability; eigenvalue; positive semidefiniteness; quantum entanglement
%\end{keyword}

%\tableofcontents

\section{Introduction}

%Quantum entanglement is a phenomenon that lies at the very core of quantum mechanics, with profound implications for our understanding of the physical universe. This concept was first brought to light by Albert Einstein, Boris Podolsky, and Nathan Rosen in their seminal 1935 paper, which introduced the EPR paradox \cite{einstein1935can}. They highlighted a peculiar aspect of quantum mechanics wherein two or more particles can become entangled, such that the state of one particle instantaneously influences the state of the other, irrespective of the distance separating them. This "spooky action at a distance," as Einstein famously termed it, appeared to challenge the classical notions of locality and realism. However, subsequent theoretical advancements and experimental validations, most notably John Bell's theorem \cite{bell1964einstein} and the experiments that ensued \cite{aspect1982experimental}, have unequivocally confirmed that entanglement is a genuine and intrinsic feature of the quantum realm, not merely a theoretical oddity.

In essence, quantum entanglement \cite{einstein1935can,bell1964einstein,aspect1982experimental,RevModPhys.81.865} occurs when particles become so intricately linked that their quantum states cannot be described independently of each other. For a pair of entangled particles, the measurement of a property (such as spin, position, or polarization) of one particle instantaneously determines the corresponding property of the other particle, even if they are separated by vast distances. This interconnectedness is quantified using the concept of non-separability in the wave function of the system. Entanglement has far-reaching implications and applications, including quantum computing \cite{PhysRevLett.90.067901,graham2022multi}, where it enables quantum bits (qubits) to perform complex computations exponentially faster than classical bits \cite{nielsen2002quantum}, and quantum cryptography \cite{PhysRevLett.96.010401,PRXQuantum.4.040305}, which promises theoretically unbreakable encryption \cite{gisin2002quantum}. Moreover, entanglement is a key resource in quantum teleportation, allowing the transfer of quantum information between distant locations without physically transmitting the particles themselves \cite{bennett1993teleporting}. Thus, quantum entanglement not only challenges our classical intuitions but also opens up new frontiers in technology and information science.

Having established the concept of entanglement, we aim to explore its detection and measurement further.
There are some entanglement criteria, such as the Schmidt decomposition \cite{miszczak2011singular}, the positive partial transposition (PPT) criterion \cite{PPT1,PPT2}, entanglement Witnesses \cite{biEW2009,MEW2013,NEW2018,norm1,Sarbicki_2008,PhysRevLett.122.120501} and the computable cross norm or realignment (CCNR) criterion \cite{Li_2011,CCNR}. It should be noted, however, that numerous other criteria exist for detecting entanglement.

%The Schmidt decomposition \cite{nielsen2002quantum,miszczak2011singular} is a fundamental concept in quantum information theory, particularly in the study of quantum entanglement. %It provides a way to express a pure quantum state of a composite system, typically consisting of two subsystems, in a simplified and insightful form.The Schmidt coefficients provide a measure of the entanglement between the two subsystems; if there is only one non-zero Schmidt coefficient, the state is separable (i.e., not entangled), whereas multiple non-zero Schmidt coefficients indicate entanglement.
%This decomposition is particularly useful because it reduces the complexity of analyzing entangled states and provides a clear quantitative measure of entanglement.

In recent years, high-dimensional entanglement has become feasible \cite{dada2011experimental,malik2016multi,erhard2020advances}, showing enhanced noise resilience compared to lower dimensions \cite{PhysRevA.71.044305,lanyon2009simplifying,PhysRevA.88.032309,mirhosseini2015high}. It is crucial to discern the successful establishment of high-dimensional entanglement from results that could be explained by low-dimensional entanglement. The challenging estimation of entanglement across various degrees of freedom involves calculating the Schmidt number \cite{PhysRevA.61.040301,PhysRevLett.132.220203}, often done through constructing a Schmidt number witness \cite{PhysRevA.63.050301}, examples of which are documented in existing literature \cite{PRA23NW,PRL24OSD,PRXQuantum.4.020324,PhysRevLett.111.030501,PhysRevX.4.011011,bavaresco2018measurements}.

The Schmidt number serves as a measure of the degree of entanglement between two subsystems in a composite quantum system. A higher Schmidt number indicates a higher degree of entanglement,
while a Schmidt number of unity implies separability with no entanglement. For mixed states, the Schmidt number is a critical metric that defines the minimal Schmidt rank of the pure states required to construct the mixed state. By analyzing the Schmidt number and associated coefficients, researchers can gain insights into the entanglement properties, structure, and
characteristics of the quantum system under consideration.

%The Schmidt number of a mixed state is a crucial metric that defines the minimal Schmidt rank of the pure states necessary for its construction. The authors \cite{PhysRevA.63.050301} introduce a canonical form of these witnesses, offering a structured framework for their implementation and developing methodologies for their optimization. Typically, the construction of these witnesses leverages the observation that quantum states with high fidelity to an entangled target state are indeed entangled. The authors \cite{PRA23NW} present an iterative algorithm to identify Schmidt number witnesses, enhancing efficiency by ensuring that the required measurements scale linearly with the system's local dimension.

The work in \cite{PhysRevA.63.050301} introduces a canonical form for Schmidt number witnesses, providing a structured framework for their implementation and optimization. These witnesses are typically constructed by observing that quantum states with high fidelity to a specific entangled target state are themselves entangled. Furthermore, the authors in \cite{PRA23NW} propose an iterative algorithm to identify Schmidt number witnesses. This method improves efficiency by ensuring that the number of required measurements scales linearly with the local dimension of the system.

The Schmidt decomposition \cite{nielsen2002quantum,miszczak2011singular} plays a central role in the characterization of pure state entanglement, in that any measure of entanglement on pure states can be written as a function of the Schmidt coefficients. This decomposition is particularly useful because it reduces the complexity of analyzing entangled states and provides a clear quantitative measure of entanglement.
The Operator Schmidt decomposition is a suitable approach to extend the notion of Schmidt coefficients to bipartite mixed states. For Schmidt number three in an arbitrary dimensional bipartite system, the authors in \cite{PRL24OSD} introduce a method to construct Schmidt number witnesses based on the Operator Schmidt decomposition. However, the construction of Schmidt number witnesses for higher than four is still an open problem in the field.

%The Schmidt decomposition [3, 4] is central to understanding pure state entanglement. It simplifies the analysis of entangled states by expressing any measure of entanglement as a function of the Schmidt coefficients. For mixed states, the Operator Schmidt decomposition extends this concept by generalizing the Schmidt coefficients to bipartite mixed states. Building on this framework, the authors in [5] present a method to construct Schmidt number witnesses for systems with a Schmidt number of three. However, developing Schmidt number witnesses for values greater than four remains an unresolved challenge, leaving a significant open question in the field of quantum entanglement.

In this Letter, we delve into characterizing Schmidt number witnesses for arbitrary dimensional bipartite quantum states with Schmidt number higher than four. Through the Operator Schmidt decomposition, we can evaluate the entanglement level of bipartite mixed quantum states based on their operator Schmidt coefficients. Subsequently,  drawing on Operator Schmidt decomposition and the theory of nonnegative matrices, we establish several Schmidt number witnesses for arbitrary dimensional bipartite quantum states with higher Schmidt number than four. Our methods not only offer viable mathematical methods for constructing high-dimensional Schmidt number witnesses in theory, but also simplify the quantification of entanglement and dimensionality.
%These findings not only offer viable mathematical methods for constructing high-dimensional Schmidt number witnesses in theory but also simplify the quantification of entanglement and dimensionality.
Consequently, we demonstrate our theoretical advancements and computational superiority by constructing Schmidt number witnesses in four- and five-dimensional systems.

\section{Preliminaries}

In this paper, we denote the set of positive integers, complex number and real number by the symbol $\mathbb{N}_{+}$, $\mathbb{C}$, $\mathbb{R}$, the set of $n$-dimensional complex (real) vectors by $\mathbb{C}^{n}$ ($\mathbb{R}^{n}$), the set of complex (real) matrices by $M_m(\mathbb{C})$ ($M_m(\mathbb{R})$), and by $S_n$ the group of permutations of the set $\{1, \ldots, n\}$. An $n$-dimensional square matrix $A$ is called positive semidefinite (PSD) if $x^{\dagger}Ax\geq 0$ for any nonzero vector $x\in \mathbb{R}^{n}$.
In the context of finite-dimensional complex Hilbert spaces, $\mathcal{H}_n$ represents such an $n$-dimensional Hilbert space, while $\mathcal{L}\left(\mathcal{H}_n\right)$ denotes the collection of linear operators acting on $\mathcal{H}_n$.
\subsection{Entanglement and witnesses }

In the field of quantum information theory, a pure quantum state \(|v\rangle\in \mathbb{C}^n\) is a unit vector, while a mixed state \(\rho\in M_n(\mathbb{C})\) is a positive semidefinite (PSD) Hermitian matrix with \(\operatorname{Tr}(\rho)=1\). A mixed state \(\rho \in M_m(\mathbb{C}) \otimes M_n(\mathbb{C})\) is called separable if there exist pure states \(\{|w_{i}\rangle \}_i\subseteq \mathbb{C}^m\) and \(\{|v_{i}\rangle \}_i\subseteq \mathbb{C}^n\) such that \(\rho=\sum\limits_{i}p_{i}|w_{i}\rangle \langle w_{i}| \otimes |v_{i}\rangle \langle v_{i}|\), with \(p_{i}\geq 0\), \(\sum p_{i}=1\).  If \(\rho\) is not separable, it is called an entangled state.

A bipartite state $\rho$ can be expressed in an operator basis as $\rho=\sum_{i j k l} \rho_{i j k l}|i\rangle\langle j|\otimes| k\rangle\langle l|$. The partial transpose with respect to the second subsystem is defined as
$ \rho^{\Gamma}=\sum_{i j k l} \rho_{i j k l}|i\rangle\langle j|\otimes(|k\rangle\langle l|)^{\top}=\sum_{i j k l} \rho_{i j k l}| i\rangle\langle j|\otimes| l\rangle\langle k|. $
The PPT criterion \cite{PPT1,PPT2} is a necessary condition for separability in bipartite systems.  If \(\rho^{\Gamma}\) has any negative eigenvalues, the state is entangled.

%The Computable Cross Norm or Realignment (CCNR) criterion \cite{Li_2011,CCNR} is used to detect entanglement by examining a matrix's cross norms. Given a density matrix \(\rho\), the criterion involves realigning \(\rho\) and computing the trace norm of this realigned matrix. If the state $\rho_{A B} \in \mathcal{S}_{A B}$ is separable,  the trace norm $\left\|\left(\rho_{A B}\right)^R\right\|_1 \leq 1$.
%This method provides a straightforward computational approach to detect certain entangled states.

For any pure state \(|\psi\rangle\) in a composite system \(A \otimes B\), it can be written as
$
|\psi\rangle = \sum_{i} \sqrt{\lambda_i} \, |u_i\rangle_A \otimes |v_i\rangle_B,
$
where \(\{\lambda_i\}\) are the Schmidt coefficients, and \(|u_i\rangle\) and \(|v_i\rangle\) are orthonormal sets. The number of non-zero Schmidt coefficients, known as the Schmidt rank (SR), provides a measure of entanglement. To extend this metric to mixed states $\rho$, the convex roof method for the Schmidt rank, known as the Schmidt number denoted by $SN(\rho)$, is utilized:
\begin{equation*}
  SN(\rho)=\min\limits_{\rho=\sum_ip_i|\psi_i\rangle\langle\psi_i|}\max SR(|\psi_i\rangle),
\end{equation*}
meaning it corresponds to the largest Schmidt rank among all composite decompositions of $\rho$, minimized across all possible decompositions. The minimization process often makes calculating the Schmidt number for a quantum state challenging. The maximally entangled state $|\phi\rangle=\frac{1}{\sqrt{n}}\sum\limits_{i=0}^{n-1}|ii\rangle$ has a Schmidt number $n$, but not $n-1$, and it can be shown that \cite{PhysRevA.61.040301,PhysRevA.59.4206}
\begin{equation}\label{d/k}
  \langle \phi^{+}|\rho_{k} |\phi^{+}\rangle \leq \frac{k}{n},\quad \forall \rho_{k} \in S_{k}.
\end{equation}

An entanglement witness is an observable that serves to distinguish between entangled and separable states. Given an entangled state $\rho$, there is an operator $W$ such that
%\begin{equation*}
  $\text{Tr}(W \sigma) \geq 0 $
%\end{equation*}
holds for all separable states $\sigma$, while $\text{Tr}(W \rho) < 0$. %One of the primary methodologies involves employing witnesses that rely on the fidelity with a specific pure target state.
These witnesses take the form:
\begin{equation}\label{fidelitywitness}
  \mathcal{W}_{\mathrm{F}}=\alpha \mathbb{I}-|\Psi\rangle \langle\Psi|,
\end{equation}
where $|\Psi\rangle$ represents a particular pure entangled target state. This witness indicates that states exhibiting a substantial fidelity with $|\Psi\rangle$, specifically those with $F_{\Psi}=\langle \Psi|\rho|\Psi\rangle >\alpha$, are entangled. The development of effective entanglement witnesses plays a pivotal role in the practical verification of entanglement in experimental setups.

Density matrices can undergo decomposition utilizing the Schmidt decomposition approach. Prior to embarking on this process, it is essential to introduce a significant definition that will be consistently referenced throughout the paper.
\begin{definition}{\rm\cite{aniello2009relation}}
Let $\rho_{A B} \in \mathcal{S}_{A B}$ be a bipartite state, and let $\left\{C_i\right\}$ and $\left\{D_i\right\}$ be HS orthonormal bases of $\mathcal{L}\left(H_A\right)$ and $\mathcal{L}\left(H_B\right)$, respectively. We define the correlation matrix $\mathcal{C}\left(\rho_{A B}\right)$ of $\rho_{A B}$ as the matrix with entries
\begin{equation}\label{Correlationmatrix}
  \mathcal{C}\left(\rho_{A B}\right)_{k l}=\operatorname{Tr}\left(C_k^{\dagger} \otimes D_l^{\dagger} \rho_{A B}\right)
\end{equation}
for $k=1, \ldots, d_A^2$ and $l=1, \ldots, d_B^2$.
\end{definition}

It is worth noticing that a change in the choice of the orthonormal bases $\left\{C_i\right\}$ and $\left\{D_i\right\}$ induces only a unitary change in the correlation matrix. This implies that the value of unitarily invariant functionals of the correlation matrix will not be affected by the choice of the local bases.

Operator Schmidt decomposition: Any bipartite state $\rho_{A B} \in \mathcal{S}_{A B}$ admits a decomposition of the form
\begin{equation}\label{OSDEW}
\rho_{A B}=\sum_{i=1}^{m} \mu_i A_i \otimes B_i
\end{equation}
where $\left\{A_i\right\}_{i=1}^{d_A^2}$ and $\left\{B_i\right\}_{i=1}^{d_B^2}$ form ( $\rho_{A B}$-dependent) bases of $\mathcal{L}\left(\mathcal{H}_A\right)$ and $\mathcal{L}\left(\mathcal{H}_B\right)$, respectively, orthonormal with respect to the $H S$ inner product. Moreover, the $\mu_i\geq 0$ are the decreasingly ordered Operator Schmidt coefficients (OSC) and $m \leq \min \left\{d_A^2, d_B^2\right\}$.

\begin{theorem}{\rm\cite{PRL24OSD}}\label{BiOSDwit}
(Bipartite OSD witnesses): Let $X$ be an operator with its OSD as in Eq. \eqref{OSDEW}. Then
\begin{equation}\label{OSDEW1}
  \mathcal{W}_{OSD}= \mu_{1}\mathbb{I}-X,
\end{equation}
is an entanglement witness for bipartite entanglement.
\end{theorem}

It is crucial to underscore that the formulation of witness \eqref{OSDEW1} encompasses the fidelity-based witness for pure states, as detailed in \eqref{fidelitywitness}, as a specific case.
This analysis provides a more comprehensive perspective, highlighting the enhanced robustness exhibited by the Operator Schmidt decomposition (OSD) witnesses.

\subsection{Schmidt number witness}

Similar to the framework of entanglement witnesses, one can define Schmidt number witnesses as observables with expectation values that are nonnegative for all states with SN $k - 1$. Consequently, a negative outcome signifies an SN of at least $k$.

As already known, states with Schmidt number $k$ form a convex set in the set of all bipartite states. These subsets are denoted $S_{k}$ and fulfill $S_{k-1} \subset S_{k}$. A $k$-Schmidt
witness detects states with Schmidt number $k$ or higher and is defined as follows:

\begin{definition}{\rm\cite{PhysRevA.63.050301}}
($k$-Schmidt witness). A hermitian operator $\mathcal{SW}_{k}$ is a $k$-Schmidt witness if and only if it fulfills
$$
\operatorname{Tr}\left(\mathcal{SW}_{k} \sigma\right) \geq 0, \quad  \forall \sigma \in S_{k-1}
$$
and there exists at least one state $\rho \in S_k$ such that
$
\operatorname{Tr}\left(\mathcal{SW}_{k} \rho\right)<0 .
$
\end{definition}

Equation \eqref{d/k} can be converted into a $(k + 1)$-Schmidt number witness through the use of the observable \cite{PhysRevA.63.050301}
$$\operatorname{Tr}\left(\mathcal{SW}_{k }\right)=\frac{k}{n}\mathbb{I}- |\phi^{+}\rangle \langle\phi^{+}|.$$

The authors \cite{PRL24OSD} demonstrate that these witnesses can be formulated in a manner similar to Equation \eqref{OSDEW1}, where the prefactor $\mu_{1}$ is substituted with a distinct value $\lambda_k$. %This parameter, while not necessarily offering an optimal constraint, acts as a limit on the overlap of pure states with Schmidt rank $k - 1$ in relation to the operator $X$.
Interestingly, these $\lambda_k$ values are determined straightforwardly as solutions to a polynomial equation of order $(k - 1)$ in the Operator Schmidt Coefficients (OSC) of hermitian operator $X$ in Equation \eqref{OSDEW1}.
However, when Schmidt numbers surpass three, the prefactor no longer retains its compact nature.

\begin{theorem}{\rm\cite{PRL24OSD}}\label{BiOSDSNW}
(OSD Schmidt number witnesses). Let $X$ be a hermitian operator, then a $k$-Schmidt witness based on the OSD is given by
\begin{equation}\label{eq016}
  \mathcal{SW}_{k }=\lambda_k \mathbb{I}-X.
\end{equation}
The coefficient $\lambda_k$ depends on the Schmidt number $k$ and can be computed by
\begin{equation}\label{eq017}
  \lambda_k = \max _{\left|\Psi^{k-1}\right\rangle \in S_{k-1}}\left\langle\Psi^{k-1}|X| \Psi^{k-1}\right\rangle,
\end{equation}
where $S_{k-1}$ denotes the state set with Schmidt number $k-1$.
\end{theorem}

For $k=3$ in the context of Schmidt number, it is demonstrated that
\begin{equation*}
  \lambda_3=  \frac{1}{2}\big(\mu_1+\mu_4+\sqrt{(\mu_1-\mu_4)^2+(\mu_2+\mu_3)^2}\big),
\end{equation*}
where $\mu_1, \mu_2,\mu_3, \mu_4$ are Operator Schmidt coefficients of density operator $X$,  and the Schmidt number witness $\mathcal{SW}_{3}=\lambda_3 \mathbb{I}-X$ identifies specifically three-dimensional entanglement.

In the case of $k=4$, the coefficient $\lambda_4$ is defined as
$$
\lambda_4=\max \left\{\operatorname{maxEig}\left(M_1^{\text {symm }}\right), \operatorname{maxEig}\left(M_2^{\text {symm }}\right)\right\},
$$
with $\operatorname{maxEig}\left(M_n^{\text {symm }}\right)$ denoting the maximum eigenvalue of the symmetrized matrices $M_1^3$ and $M_2^3$, given by
\begin{equation}\label{matricesM42}
M_1^3=\left(\begin{array}{ccc}
\mu_1 & \mu_2 & \mu_4 \\
\mu_3 & \mu_6 & \mu_7 \\
\mu_5 & \mu_8 & \mu_9
\end{array}\right), \quad M_2^3=\left(\begin{array}{ccc}
\mu_1 & \mu_2 & \mu_5 \\
\mu_3 & \mu_4 & \mu_7 \\
\mu_6 & \mu_8 & \mu_9
\end{array}\right) .
\end{equation}

Analysis reveals that while $\lambda_4$ is computable through numerical methods, %its functional representation involving variables $\mu_1$ to $\mu_9$ is notably intricate,
defying precise mathematical expression.
For Schmidt numbers greater than four, numerous Operator Schmidt coefficients are derived from the product of vector Schmidt coefficients $s_{\alpha}s_{\beta}$ and are arranged in decreasing order, expanding the possibilities for defining matrices. As the Schmidt number increases, the computation of eigenvalues for multiple matrices experiences a significant upsurge, leading to a rise in computational expenses and complexity with dimensionality.

\subsection{Nonnegative matrix spectral radius theory}

Given a nonnegative matrix $M \in \mathbb{M}_n$, $P_r$ represents the sum of entries in the $r$-th row of $M$, $p=\min _r P_r$, $P=\max _r P_r$, and $\rho(M)$ denotes the maximum eigenvalue of $M$. Frobenius \cite{frobenius1908matrizen} demonstrated the following:
\begin{equation}\label{NMB1}
 p \leq \rho(M) \leq P,
\end{equation}
Moreover, Ledermann \cite{ledermann1950bounds} established that:
\begin{equation}\label{NMB2}
p+m\left(\frac{1}{\sqrt{\delta}}-1\right) \leq \rho(M) \leq P-m(1-\sqrt{\delta})
\end{equation}
where $m=\min _{r, s}\left(M_{r s}\right)$ is the smallest entry in $M$ and $\delta=\max\limits _{P_r<P_s}\left(\frac{P_r}{P_s}\right)$. Ostrowski \cite{ostrowski1952bounds} later refined the inequalities in \eqref{NMB2}: if $p<P$, then
\begin{equation}\label{NMB3}
p+m\left(\frac{1}{\sigma}-1\right) \leq \rho(M) \leq P-m(1-\sigma)
\end{equation}
where $\sigma=\sqrt{\frac{(p-m)}{(P-m)}}$. Subsequently, Brauer \cite{brauer1957theorems} improved upon the bounds in \eqref{NMB3} with:
\begin{equation}\label{NMB4}
p+m(h-1) \leq \rho(M) \leq P-m\left(1-\frac{1}{g}\right)
\end{equation}
where $g=\frac{P-2 m+\sqrt{P^2-4 m(P-p)}}{2(p-m)}$ and $h=\frac{-p+2 m+\sqrt{P^2+4 m(P-p)}}{2 m}$.

Inequalities \eqref{NMB1}-\eqref{NMB4} above illustrate that the overall approximation of an individual nonnegative matrix is contingent on its internal elements.
%Nonetheless, our research emphasizes determining the maximum spectral radius value across a specific set of matrices.

\section{Characterizing Schmidt number witnesses for bipartite quantum states in arbitrary dimensions}

%This section investigates the bipartite quantum state of the Schmidt number witness in general case. %Characterizing High Schmidt Number Witnesses in Arbitrary Dimensions System
This section characterizes high Schmidt number witnesses in arbitrary dimensions bipartite quantum system.
%We will now shift our attention to identifying witnesses that certify the Schmidt number in arbitrary dimensions.
Our goal is to ascertain the coefficient $\lambda_{k+1} (k\geq 5)$ that ensures the Schmidt number witness produces a nonnegative expectation value for all states $\left|\psi_k\right\rangle$ with $\mathrm{SN} (k)$.
It is known \cite{PRL24OSD} that determining the precise value of $\lambda_k (k\geq 4)$ becomes highly challenging.

To address this limitation, we introduce  upper bounds for the Schmidt number witness coefficient $\lambda_k$ based on the theory of nonnegative matrix spectral radius. Subsequently, we successfully construct a Schmidt number witness in arbitrary dimensions, which depends on the Operator Schmidt coefficients $\mu_1,\cdots, \mu_{(k-1)^2}$.  These insights are available for arbitrary dimensions, enabling the construction of Schmidt number witnesses across different dimensional spaces.

There is an important fact about the Operator Schmidt coefficients, that is, the sum of the squares of the OSC of any density matrices equals its purity \cite{lupo2008bipartite,aniello2009relation}. Hence, the following Lemma is desired immediately.

\begin{lemma}\label{lem3.1mu}
Let a bipartite state $\rho_{A B} \in \mathcal{S}_{A B}$ admit an Operator Schmidt decomposition of the form $\rho_{A B}=\sum\limits_{i=1}^{m} \mu_i A_i \otimes B_i$,  where $\left\{A_i\right\}_{i=1}^{d_A^2}$ and $\left\{B_i\right\}_{i=1}^{d_B^2}$ form  bases of $\mathcal{L}\left(\mathcal{H}_A\right)$ and $\mathcal{L}\left(\mathcal{H}_B\right)$, respectively, orthonormal with respect to the $H S$ inner product. Then all Operator Schmidt coefficients $\mu_i$ are less than one, i.e,
\begin{equation*}
  \mu_{i}\leq 1, \quad i=1,\cdots, m.
\end{equation*}
\end{lemma}
\begin{proof}
%For Operator Schmidt coefficients and Schmidt equivalence classes of bipartite states,
The following facts are known in the literature \cite{lupo2008bipartite,aniello2009relation}:
\begin{equation*}
  \sum\limits_{i}\mu_{i}^2=\|\langle \rho_{A B},\rho_{A B}\rangle\|_{HS}=\|\rho_{A B}\|_{2}^{2}=\operatorname{Tr}(\rho_{A B}^2)\leq 1.
\end{equation*}
This shows all Operator Schmidt coefficients $\mu_i$ are less than one. The proof is completed.
\end{proof}

From Theorem \ref{BiOSDSNW}, equation \eqref{eq017} ensures $\mathcal{SW}_{k+1}$'s expectation positivity on states with Schmidt number $k$. Subsequently, we analyze \eqref{eq016} to determine coefficients $\lambda_{k+1}$:
\begin{equation}\label{eq018}
\max _{\left|\Psi^{k}\right\rangle \in S_{k}}\left\langle\Psi^{k}|X| \Psi^{k}\right\rangle=\max _{\left|\Psi^{k}\right\rangle \in S_{k}} \operatorname{Tr}\left(X\left|\Psi^{k}\right\rangle\left\langle\Psi^{k}\right|\right) .
\end{equation}
By decomposing $X$ and $\left|\Psi^{k}\right\rangle\left\langle\Psi^{k}\right|$, we find:
\begin{eqnarray}\label{eq019}
&&\max_{\bar{s}_l}\operatorname{Tr}\left(\sum_j \mu_j G_j^A \otimes G_j^B \sum_l \tilde{s}_l H_l^A \otimes H_l^B\right) \nonumber\\
\leq &&\max_{\bar{s}} \sum_{j,l} \mu_j \tilde{s}_{l\downarrow}.
\end{eqnarray}
Here, $\tilde{s}_{l\downarrow}$ are the decreasingly ordered operator Schmidt coefficients of $\left|\Psi^{k}\right\rangle\left\langle\Psi^{k}\right|$, given by the products of vector Schmidt coefficients $s_i$ of $\left|\Psi^{k}\right\rangle = \sum_{i=1}^{k} s_i|\alpha_i \beta_i\rangle$.
Furthermore, the coefficients $s_{1},\cdots, s_{k}$ obey the constraints $s_{1}^2+\cdots+ s_{k}^2=1$ and $s_{1}\geq  \cdots \geq  s_{k}\geq 0$.
%In this context, these $\lambda_{k+1}$ are simply given by the solution of a $k$th order polynomial equation in the OSC of $X$.
In this scenario, the values of $\lambda_{k+1}$ are determined as solutions to a polynomial equation of order $k$ within the OSC associated with $X$.

Following a procedure akin to that of $k=3,4$ \cite{PRL24OSD}, the Operator Schmidt decomposition (OSD) for the projector $\left|\psi_k\right\rangle\left\langle\psi_k\right|$ is defined by the set
\begin{equation*}
  \left\{\tilde{\mu}_j\right\}=\left\{s_1 s_1, \cdots, s_1 s_k, s_2^2, \cdots, s_{k-1} s_k, s_k^2\right\},
\end{equation*}
For the function $f_{i}\left(s_1, \cdots, s_k\right)=\sum_{i j} \mu_i \tilde{\mu}_j$ to be derived, it is necessary to arrange the OSD coefficients $\tilde{\mu}_j$ in descending order.
The coefficient $\lambda_{k+1}$ can be obtained as the maximum of the following functions:
\begin{equation}\label{eqfk2}
  \max_{s_1, \cdots, s_k}\left\{f_i\left(s_1, \cdots, s_k\right)\right\},
\end{equation}
where $i$ depends on the number of arrangements for $\{\tilde{\mu}_j\}$.

There are always such orderings for $\{s_{i} s_{j}\}_{1\le i,j\le k}$:
\begin{equation*}
  s_{1}^{2}\geq s_{1} s_{2}=s_{2} s_{1}\geq \cdots\geq  s_{1}s_{m} \geq  s_{2}^{2}\geq s_{2}s_{3}\geq  \cdots\geq s_{k}^{2},
\end{equation*}
there are corresponding $k$ order square matrix $M_u^{k}$ as following:
\begin{equation*}
M_u^{k}=
  \begin{pmatrix}
   \mu_1 & \mu_2 & \cdots & \mu_{2k-4}&\mu_{2k-2}\\
    \mu_3 & \mu_{2k} & \cdots & \mu_{4k-5} &\mu_{4k-3} \\
    \vdots & \vdots & \vdots & \vdots &\vdots\\
    \mu_{2k-3}& \mu_{4k-6} & \cdots & \mu_{k^2-3} & \mu_{k^2-2} &\\
    \mu_{2k-1} & \mu_{4k-4} & \cdots & \mu_{k^2-1} &\mu_{k^2}\\
  \end{pmatrix},
\end{equation*}
where $u$ is the number of arrangements for $\{\tilde{\mu}_j\}$.
It is not difficult to find that the first and k-th row of matrix $M_u^{k}$ exhibit the largest sum of the row values.

Based on Lemma \ref{lem3.1mu} and inequality \eqref{NMB4}, we present our results and construct suitable Schmidt number witnesses.

\begin{theorem}\label{thmthetamn}
Consider a hermitian operator $X$ in $\mathcal{L}(\mathcal{H}_m) \otimes \mathcal{L}(\mathcal{H}_n)$ $(n\geq m\geq 5)$, the optimal $(k+1)$-Schmidt witness based on the OSD is $\mathcal{SW}_{(k+1) }=\lambda_{k+1} \mathbb{I}-X$, with $\lambda_{k+1} = \max _{\left|\Psi^{k}\right\rangle \in S_{k}}\left\langle\Psi^{k}|X| \Psi^{k}\right\rangle$ and $S_{k}$ denotes the states set with Schmidt number $k$. We define the $\theta_{k+1}$ as following:
\begin{equation*}
  \theta_{k+1}= P-\mu_{k^2}+\mu_{k^2}\frac{2(p-\mu_{k^2})}{P-2 \mu_{k^2}+\sqrt{P^2-4 \mu_{k^2}(P-p)}},
\end{equation*}
where $\mu_i, i=1,\cdots,k^2$ are Operator Schmidt coefficients of $X$,  $P=\mu_1+\mu_2+\dots+\mu_{2k-2}$ and $p=\mu_{2k-1}+\mu_{4k-4}+\cdots+\mu_{k^2-1}+\mu_{k^2}$.
Notably,
%\begin{tiny}
%\begin{equation}\label{eqsnm3}
%  \lambda_{k+1}\leq \theta_{k+1}= P-\mu_{k^2}+\mu_{k^2}\frac{2(p-\mu_{k^2})}{P-2 \mu_{k^2}+\sqrt{P^2-4 \mu_{k^2}(P-p)}},
%\end{equation}
%\end{tiny}
\begin{equation}\label{eqsnm3}
  \lambda_{k+1}\leq \theta_{k+1}.
\end{equation}
Subsequently, we construct the following $(k+1)$-Schmidt witness:
\begin{equation*}
  \mathcal{SW}_{(k+1)}=\theta_{k+1} \mathbb{I}-X.
\end{equation*}
\end{theorem}

While $\lambda_{k+1}$ represents the optimal witness coefficient for Schmidt number $k$, its direct determination or calculation is often challenging. Hence, we explore its feasible upper limit to facilitate the construction of an appropriate witness.

The proof of Theorem \ref{thmthetamn} is presented in Appendix A. It is evident that $\eta_{k+1}$ is solely dependent on the Operator Schmidt coefficients, with $\eta_{k+1}\leq P_{k+1}$. In cases where the smallest operator Schmidt coefficient $\mu_{k^2}=0$, we have $\theta_{k+1}=P_{k+1}=\mu_1+\mu_2+\dots+\mu_{2k-2}$. %Furthermore, the utilization of inequality \eqref{NMB3} enables us to promptly present the subsequent conclusion.

In a similar manner, based on Lemma \ref{lem3.1mu} and inequality \eqref{NMB2}, there is the following result.

\begin{proposition}\label{thmetamn}
Using the same notations as in Theorem \ref{thmthetamn}.  We define the $\eta_{k+1}$ as following:
\begin{eqnarray*}
  \eta_{k+1} = && \mu_1+\mu_2+\dots+\mu_{2k-2} \\
   -&& \mu_{k^2}\left(1-\sqrt{\frac{\mu_{2k-1}+\mu_{4k-4}+\cdots+\mu_{k^2-1}+\mu_{k^2}}{\mu_1+\mu_2+\dots+\mu_{2k-2}}}\right).
\end{eqnarray*}

Subsequently, we construct the following $(k+1)$-Schmidt number witness:
\begin{equation*}
  \mathcal{SW}_{(k+1) }=\eta_{k+1} \mathbb{I}-X.
\end{equation*}
\end{proposition}

Similarly, by referencing inequality \eqref{NMB3}, we can promptly introduce the subsequent conclusion.

\begin{proposition}\label{thmzetamn}
Using the same notations as in Theorem \ref{thmthetamn}.  We introduce the $\zeta_{k+1}$ as following:
\begin{eqnarray*}
  \zeta_{k+1} = && \mu_1+\mu_2+\dots+\mu_{2k-2} \\
   -&& \mu_{k^2}\left(1-\sqrt{\frac{\mu_{2k-1}+\mu_{4k-4}+\cdots+\mu_{k^2-1}}{\mu_1+\mu_2+\dots+\mu_{2k-2}-\mu_{k^2}}}\right).
\end{eqnarray*}

Next, we formulate the $(k+1)$-Schmidt number witness as follows:
\begin{equation*}
  \mathcal{SW}_{(k+1) }=\zeta_{k+1} \mathbb{I}-X.
\end{equation*}
\end{proposition}

The proof of Proposition \ref{thmetamn} and Proposition \ref{thmzetamn} are presented in Appendix B.

\begin{remark}
It is worth noting that the following inequality holds:
\begin{equation}\label{eq019all}
  \lambda_{k+1}\leq \theta_{k+1}\leq \zeta_{k+1}\leq \eta_{k+1} \leq \mu_1+\mu_2+\dots+\mu_{2k-2},
\end{equation}
indicating that the coefficient $\lambda_{k+1}$ is less than the largest row sum of all matrices $M_{i}^{k}$.
%Specifically, when $\mu_{k^2}=0$, some of the aforementioned equations hold: $\theta_{k+1}= \zeta_{k+1}=\eta_{k+1} = \mu_1+\mu_2+\dots+\mu_{2k-2}$.
\end{remark}

Specifically, when $\mu_{k^2}=0$, some of the aforementioned equations hold.%: $\theta_{k+1}= \zeta_{k+1}=\eta_{k+1} = \mu_1+\mu_2+\dots+\mu_{2k-2}$.

\begin{proposition}\label{pro01}
Using the same notations as in Theorem \ref{thmthetamn}, Remark \ref{thmetamn} and Remark \ref{thmzetamn}.
If $\mu_{k^2}=0$, there are
$$\lambda_{k+1}\leq \theta_{k+1}= \zeta_{k+1}=\eta_{k+1} = \mu_1+\mu_2+\dots+\mu_{2k-2}.$$
\end{proposition}

Proposition \ref{pro01} reveals that when the smallest operator Schmidt coefficient $\mu_{k^2}$ approaches 0, both the optimal witness coefficient $\lambda_{k+1}$ and the four upper bounds $\theta_{k+1}, \zeta_{k+1}, \eta_{k+1},P_{k+1}$ exhibit close proximity. Moreover, in the scenario where $\mu_{k^2}=0$, the optimal witness coefficient aligns precisely with all four upper bounds.
%\subsection{Characterizing Schmidt number witnesses for bipartite quantum states in four dimension}
%
%In general, the precise mathematical expression for the coefficient $\lambda_4$ in the presence of the operator Schmidt coefficients $\mu_1, \ldots, \mu_{9}$ eludes us.
%It is worth noting that the matrix set $\{M_i\}$ comprises solely two matrices for $k=4$.
%Consequently, we contemplate employing the nonnegative matrix spectral radius theory to establish  upper bounds for the coefficient $\lambda_4$ in Theorem \ref{thmetamn}, \ref{thmzetamn} and \ref{thmthetamn}, followed by constructing the corresponding Schmidt number witness $\mathcal{W}_{4 \mathrm{S}}$.

%Firstly, we investigate a quantum state $\rho_0$ shown in \cite{PhysRevLett.124.200502,PRL24OSD} whose Schmidt number can not be certified with fidelity-based witnesses, but with OSD witnesses.
Firstly, we examine the quantum state $\rho_0$ as presented in \cite{PhysRevLett.124.200502,PRL24OSD}. The Schmidt number of this state cannot be certified using fidelity-based witnesses, but it can be verified with OSD witnesses.

\begin{example}
We consider the mixed two-ququad state
$$\rho_0=\frac{1}{2}|\phi^{+}_{3}\rangle \langle \phi^{+}_{3}|+\frac{1}{4}(|23\rangle+|32\rangle)(\langle 23|+\langle 32|),$$
with $|\phi^{+}_{3}\rangle=\frac{1}{\sqrt{3}}(|00\rangle+|11\rangle+|22\rangle).$
\end{example}

The state under consideration is characterized by Schmidt number three. In this scenario, fidelity-based Schmidt witnesses exclusively identify Schmidt number two \cite{PhysRevLett.124.200502}, serving as an indicator for entanglement verification.

Based on the operator Schmidt decomposition (OSD), there are the following operator Schmidt coefficients(OSC):
$\mu_1=0.3469, \mu_2=\mu_3=0.2500, \mu_4=0.1802, \mu_5=\cdots=\mu_{12}=0.1667,\mu_{13}=\cdots=\mu_{16}=0.$

Therefore, based on Proposition \ref{pro01}, there are
\begin{equation*}
  \lambda_{5}(\rho_0)=0.6848 \leq \theta_{5}(\rho_0)= \zeta_{5}(\rho_0)=\eta_{5}(\rho_0)=0.9438.
\end{equation*}
In fact, the optimal Schmidt number witness are as following:
\begin{equation*}
  \mathcal{SW}_{\rho_0 }= 0.6848*\mathbb{I}-\rho_0.
\end{equation*}

Generally, accurately determining the maximum value of spectral radii for a set of nonnegative matrices is a challenging and expensive task. Conversely, establishing the upper bounds is more straightforward. For instance, in Example 1, we must compute the spectral radii of 12 nonnegative matrices. The specific matrix arrangement is documented in Appendix D.

\begin{example}
We consider the mixed two-qudit state
\begin{align*}
  \rho_k= & \frac{1}{2}|\phi^{+}\rangle \langle \phi^{+}|+\frac{1}{4}\big(|k-2, k-1\rangle \\
   & +|k-1,k-2\rangle\big)\big(\langle k-2, k-1|+\langle k-1,k-2|\big),
\end{align*}
%$\rho_k=\frac{1}{2}|\phi^{+}\rangle \langle \phi^{+}|+\frac{1}{4}(|k-2, k-1\rangle+|k-1,k-2\rangle)(\langle k-2, k-1|+\langle k-1,k-2|)$
with $|\phi^{+}\rangle=\frac{1}{\sqrt{k}}\sum\limits_{i=0}^{k-1}|ii\rangle.$
\end{example}

Based on the Operator Schmidt decomposition (OSD), If $k=2$, all the Operator Schmidt coefficients(OSC) are:
$\mu_1=0.2571, \mu_2=0.2517,  \mu_3=0.2473, \mu_4=0.2439.$

If $k=3$, all the Operator Schmidt coefficients(OSC) are:
$\mu_1=\mu_2=0.4167,  \mu_3=\mu_4=\mu_5=\mu_6=\mu_7=0.1667,  \mu_8=\mu_9=0.0833.$

If $k=4$, all the Operator Schmidt coefficients(OSC) are:
$\mu_1=\mu_2=0.3750,   \mu_3=\cdots=\mu_{16}=0.1250.$

If $k=5$, all the Operator Schmidt coefficients(OSC) are:
$\mu_1=\mu_2=0.3500,  \mu_3=\mu_4=0.1500,  \mu_5=\cdots=\mu_{25}=0.1000.$

Table 1 shows the coefficient $\lambda_{k+1}$ and upper bounds $\theta_{k+1},\zeta_{k+1}, \eta_{k+1}, P_{k+1}$ for the construction of Schmidt number witnesses for $\rho_k$ in Example 1.
It is not difficult to find that $\rho_k$ is a mixed state with fixed Schmidt number $k$.
In general, minimizing over pure states with a fixed Schmidt rank continues to pose a significant challenge, often lacking analytically tractable solutions in numerous instances \cite{PhysRevA.63.050301,PRA23NW}.
Our findings, rooted in Operator Schmidt decomposition and nonnegative matrix spectral theory, enable the establishment of an upper bound on the optimal witness coefficients, facilitating the construction of the associated Schmidt number witness.

\begin{table}[h]\label{table}
\caption{The coefficient $\lambda_{k+1}$ and upper bounds $\theta_{k+1},\zeta_{k+1}, \eta_{k+1}, P_{k+1}$ for construction of Schmidt number witnesses for $\rho_k$.}
%\begin{center}
\begin{footnotesize}
\begin{tabular}{|c|c|c|c|c|c|}
  \hline
   Mixed states \qquad&  $\lambda_{k+1}$ \qquad&  $\theta_{k+1}$ \qquad &  $\zeta_{k+1}$ \qquad &  $\eta_{k+1}$ \qquad&  $P_{k+1}$\qquad\\
  \hline
  $\rho_{2}, (k=2)$   & 0.5001 & 0.5002 & 0.5006 & 0.5045 & 0.5088\\
  \hline
  $\rho_{3}, (k=3)$   & 0.6642 & 0.9420 & 0.9603 & 0.9649 & 1.0001\\
  \hline
  $\rho_{4}, (k=4)$  & 0.6545 & 0.9330 &0.9568& 0.9634 & 1.0000\\
  \hline
  $\rho_{5}, (k=5)$  & 0.6628 & 0.9447 &0.9649& 0.9690 & 1.0500\\
  \hline
\end{tabular}
\end{footnotesize}
%\end{center}
\end{table}

%\begin{example}
%We consider the mixed two-ququad state
%$\rho_0=\frac{1}{2}|\phi\rangle \langle \phi|+\frac{1}{4}(|23\rangle+|32\rangle)(\langle 23|+\langle 32|)$
%with $|\phi\rangle=\frac{1}{\sqrt{3}}(|00\rangle+|11\rangle+|22\rangle).$
%\end{example}
%Based on the Operator Schmidt decomposition (OSD), there are following the operator Schmidt coefficients(OSC):
%$\mu_1=0.3469, \mu_2=\mu_3=0.2500, \mu_4=0.1802, \mu_5=\cdots=\mu_{12}=0.1667,\mu_{13}=\cdots=\mu_{16}=0.$
%
%
%\begin{table}[h]\label{table}
%\caption{Construction of Schmidt number witnesses for Example 2 and randomly generated mixed states with some specific Schmidt rank }
%\begin{center}
%\begin{tabular}{|c|c|c|c|c|c|}
%  \hline
%   Mixed states \qquad&  $\lambda_{k+1}$ \qquad& $\theta_{k+1}$ \qquad &  $\zeta_{k+1}$ \qquad &  $\eta_{k+1}$ \qquad& $P_{k+1}$\qquad\\
%  \hline
%   $\rho_{0}, (k=3)$  & 0.6848 & 0.9438 &0.9438& 0.9438 & 0.9438\\
%  \hline
%  $\rho_{4}, (k=5)$, $SR=9$  & 0.4041 & 0.7250 &0.7250& 0.7250 & 0.7250\\
%  \hline
%  $\rho_{5}, (k=5)$, $SR=16$  & 0.3312 & 0.5960 &0.5960& 0.5960 & 0.5960\\
%  \hline
%\end{tabular}
%\end{center}
%\end{table}

We investigate different mixed states by creating 50 instances using MATLAB. Each mixed state is composed of 10,000 convex combinations of pure states, which avoids the situation where the four upper bounds overlap as in Proposition \ref{pro01}. Our analysis covers scenarios with Schmidt number $k$ values of 2, 3, 4, and 5, that is, the arbitrary dimensional bipartite systems are Schmidt number 2, 3, 4, and 5 respectively. The detailed numerical results are shown in Figure 1. It can be seen from Figure 1 that the four upper bounds $\theta_{k+1},\zeta_{k+1}, \eta_{k+1}, P_{k+1}$ gradually approach each other as the system dimension increases.
Moreover, increasing the number of pure states of mixed convex combinations does not change this phenomenon.
In addition, Figure 1 shows that the optimal witness coefficient $\lambda_{k+1}$ seems to converge to $\frac{1}{k}$, which means we can construct the following Schmidt number witness:
\begin{equation}\label{figwit}
  \mathcal{SW}_{(k+1) }=(\frac{1}{k}+\delta) \mathbb{I}-X, \quad  \delta>0.
\end{equation}
In other words, when the quantity of composed mixed states significantly surpasses the system's dimension, a randomly generated mixed state approximates the maximum mixed state. Consequently, the coefficient observed via the optimal Schmidt number converges towards $\frac{1}{k}$.

%It is noting that if $X=\frac{1}{\sqrt{k}}\sum\limits_{i=0}^{k-1}|ii\rangle$, the OSC are 1,0($k^2-1)$ in our observable.

\begin{figure}[ht]%\label{Figure1}
    \centering
    \begin{subfigure}{0.48\linewidth}
        \centering
        \includegraphics[width=\linewidth]{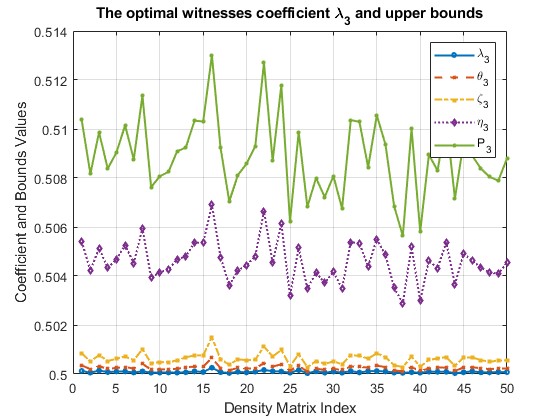}
        \caption{The witnesses coefficient $\lambda_{3}$ and upper bounds}
    \end{subfigure}
    \hfill
    \begin{subfigure}{0.48\linewidth}
        \centering
        \includegraphics[width=\linewidth]{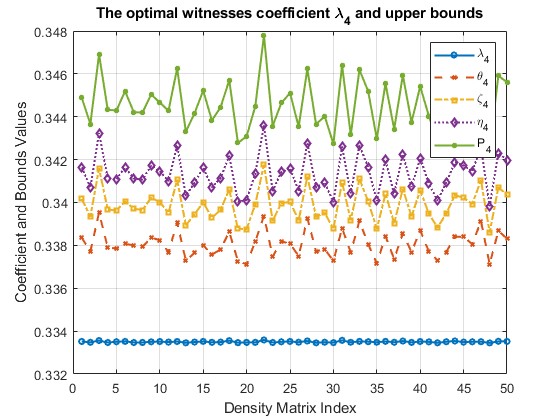}
        \caption{The witnesses coefficient $\lambda_{4}$ and upper bounds}
    \end{subfigure}\\
    \begin{subfigure}{0.48\linewidth}
        \centering
        \includegraphics[width=\linewidth]{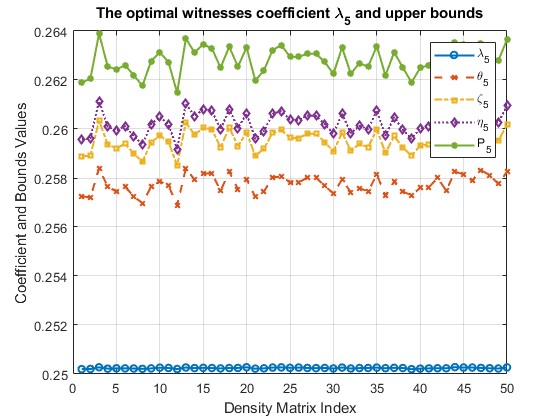}
        \caption{The witnesses coefficient $\lambda_{5}$ and upper bounds}
    \end{subfigure}
    \hfill
    \begin{subfigure}{0.48\linewidth}
        \centering
        \includegraphics[width=\linewidth]{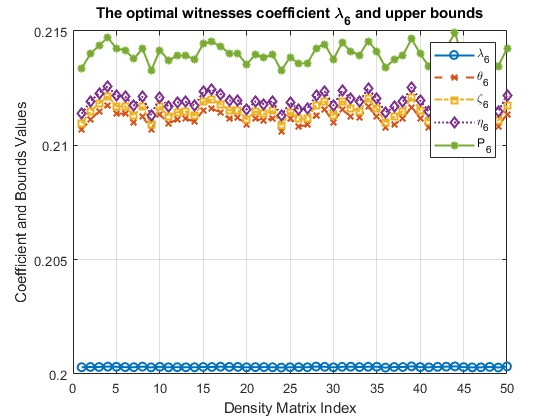}
        \caption{The witnesses coefficient $\lambda_{6}$ and upper bounds}
    \end{subfigure}
    \caption{The optimal witnesses coefficient $\lambda_{k+1}$ and upper bounds $\theta_{k+1},\zeta_{k+1}, \eta_{k+1}, P_{k+1}$ of the Schmidt number witnesses associated with a mixed state generated randomly in MATLAB are determined. The four images (a), (b), (c), and (d) depict 50 randomly generated mixed states for Schmidt number $k=2, 3, 4,$ and $5,$ respectively. Each mixed state is produced through 10,000 convex combinations of pure states.}
\end{figure}

%On the other hand, we consider Schmidt number witnesses for higher Schmidt number, such as higher dimensional bipartite  systems system with Schmidt number 10. The results of  correlation analysis with deferent number combinations of pure states  can be compared in Figure 2. It is apparent from (a) and (b) in Figure 2 that if the quantity of pure states composing the mixed state is equal to or below the square of the system dimension, the four upper bounds $\theta_{11},\zeta_{11}, \eta_{11}, P_{11}$ will exhibit consistency, corroborated by the findings in Proposition \ref{pro01}. If the number of pure states constituting the mixed state significantly exceeds the square of the system dimension, such as (c) (d) in Figure 2, the four upper bounds $\theta_{11},\zeta_{11}, \eta_{11}, P_{11}$ will sequentially emerge, nearing the optimal witness coefficient $\lambda_{11}$ at this juncture. In this situation, descriptions in (c) and (d) in Figure 2 seem to verify the feasibility of Schmidt number witness in \eqref{figwit}.

On the other hand, we consider Schmidt number witnesses for higher Schmidt numbers, such as higher-dimensional bipartite systems with a Schmidt number of 10. The results of the correlation analysis with different combinations of pure states are presented in Figure 2. It is evident from Figures 2(a) and 2(b) that if the number of pure states composing the mixed state is equal to or less than the square of the system dimension, the four upper bounds $\theta_{11}, \zeta_{11}, \eta_{11}, P_{11}$ exhibit consistency, as corroborated by Proposition \ref{pro01}. However, if the number of pure states constituting the mixed state significantly exceeds the square of the system dimension, as shown in Figures 2(c) and 2(d), the four upper bounds $\theta_{11}, \zeta_{11}, \eta_{11}, P_{11}$ sequentially approach the optimal witness coefficient $\lambda_{11}$. In this context, the descriptions in Figures 2(c) and 2(d) appear to verify the feasibility of the Schmidt number witness in \eqref{figwit}.

\begin{figure}[ht]
    \centering
    \begin{subfigure}{0.48\linewidth}
        \centering
        \includegraphics[width=\linewidth]{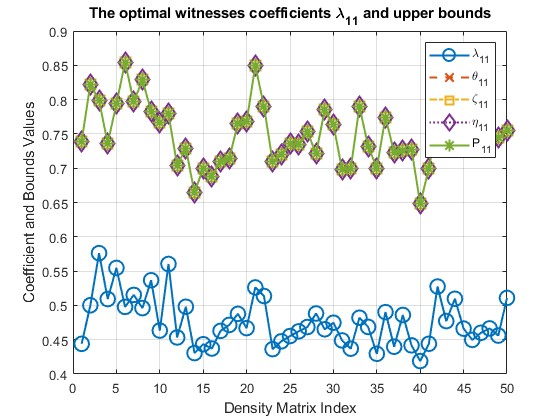}
        \caption{The witness coefficients $\lambda_{11}$ and upper bounds of the mixed state are derived from combinations of 5 pure states.}
    \end{subfigure}
    \hfill
    \begin{subfigure}{0.48\linewidth}
        \centering
        \includegraphics[width=\linewidth]{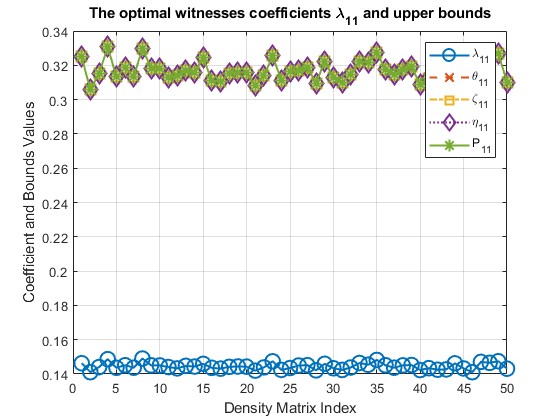}
        \caption{The witness coefficients $\lambda_{11}$ and upper bounds of the mixed state are derived from combinations of 100 pure states.}
    \end{subfigure}\\
    \begin{subfigure}{0.48\linewidth}
        \centering
        \includegraphics[width=\linewidth]{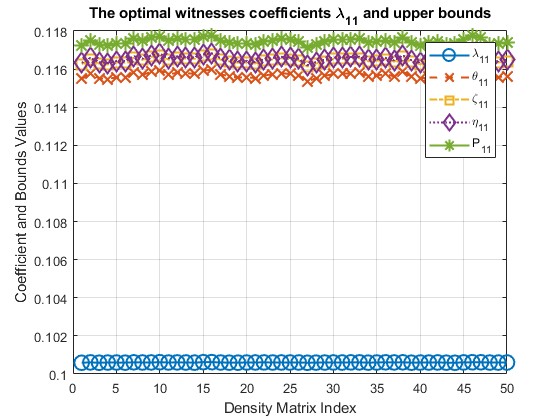}
        \caption{The witness coefficients $\lambda_{11}$ and upper bounds of the mixed state are derived from combinations of 10000 pure states.}
    \end{subfigure}
    \hfill
    \begin{subfigure}{0.48\linewidth}
        \centering
        \includegraphics[width=\linewidth]{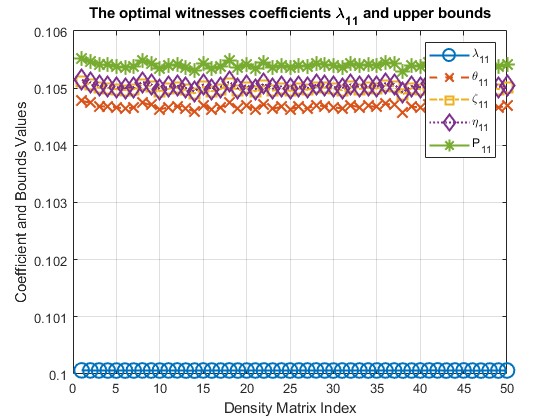}
        \caption{The witness coefficients $\lambda_{11}$ and upper bounds of the mixed state are derived from combinations of 100000 pure states.}
    \end{subfigure}
    \caption{The optimal witnesses coefficients $\lambda_{11}$ and upper bounds $\theta_{11}$, $\zeta_{11}$, $\eta_{11}$, and $P_{11}$ for mixed bipartite states with the Schmidt numbers $10$  are obtained through random generation using MATLAB. As can be seen from the Figure 2, the coefficients $\lambda_{11}$ and the upper bounds $\theta_{11}$, $\zeta_{11}$, $\eta_{11}$, and $P_{11}$ of 50 randomly generated mixed states. Descriptions in Figures 2(a), Figures 2(b), Figures 2(c), and Figures 2(d) correspond that each mixed state is created through convex combinations of 5, 100, 10000, and 100000 pure states, respectively.}
\end{figure}

%The optimal coefficients $\lambda_{11}$ and upper bounds $\theta_{11}$, $\zeta_{11}$, $\eta_{11}$, and $P_{11}$ of the Schmidt numbers corresponding to the $10\times 10$ mixed states randomly generated by MATLAB. (a), (b), (c), and (d) correspond to the descriptions of the coefficients $\lambda_{11}$ and upper bounds $\theta_{11}$, $\zeta_{11}$, $\eta_{11}$, $P_{11}$ of 50 randomly generated mixed states, respectively. Each mixed state is generated by 5, 50, 100, and 10,000 pure convex combinations.

\section{Conclusion}
Understanding quantum entanglement is pivotal for the advancement of quantum technologies. Evaluating entanglement levels involves counting the entangled degrees of freedom, leading to the derivation of the Schmidt number.
This study introduces efficient analytical tools for characterizing high Schmidt number witnesses in bipartite quantum states across various dimensions. Our approach not only provides practical mathematical techniques for crafting high-dimensional Schmidt number witnesses but also streamlines entanglement and dimensionality quantification. Noteworthy is our development of high-dimensional Schmidt number witnesses in arbitrary dimensional bipartite quantum systems, relying solely on operator Schmidt coefficients. We showcase our theoretical progress and computational superiority by constructing Schmidt number witnesses in bipartite quantum systems of dimensions four and five.

%We have developed an innovative mathematical method to characterize high Schmidt number witnesses in arbitrary dimensional bipartite quantum systems. These results enable the construction of Schmidt number witnesses in dimensions 4 and higher. The coefficients of these witnesses are determined solely by the maximum value of a multivariate quadratic polynomial based on the operator Schmidt coefficients, specifically the spectral radius of the set of nonnegative matrices formed by these coefficients. Through the analysis of these nonnegative matrices, we establish the upper bounds for their spectral radius, facilitating the creation of the corresponding Schmidt number witnesses. Our approach not only offers a theoretically plausible mathematical technique for building high-dimensional Schmidt number witnesses but also streamlines entanglement quantization and dimensionality assessment.

\begin{acknowledgments}
Research of Sze was supported by a HK RGC grant PolyU 15300121 and a PolyU research grant 4-ZZRN. The HK RGC grant also supported the post-doctoral fellowship of Xiong at the Hong Kong Polytechnic University.
\end{acknowledgments}

%\section*{References}

\bibliography{SNaps}% Produces the bibliography via BibTeX.

%\appendix

\section{Appendix}

In the Appendix, we will feature the proofs of Theorem \ref{thmthetamn}, Proposition \ref{thmetamn} and Proposition \ref{thmzetamn}, along with the presentation of our findings in arbitrary dimensional bipartite quantum system with Schmidt number four and five.

\begin{widetext}

\subsection{Proof of Theorem \ref{thmthetamn}}
From Theorem \ref{BiOSDSNW}, equation \eqref{eq017} ensures $\mathcal{W}_{k+1}$'s expectation positivity on states with Schmidt number $k$. Subsequently, we analyze \eqref{eq016} to determine coefficients $\lambda_{k+1}$:
\begin{equation*}%\label{eq018}
\max _{\left|\Psi^{k}\right\rangle \in S_{k}}\left\langle\Psi^{k}|X| \Psi^{k}\right\rangle=\max _{\left|\Psi^{k}\right\rangle \in S_{k}} \operatorname{Tr}\left(X\left|\Psi^{k}\right\rangle\left\langle\Psi^{k}\right|\right) .
\end{equation*}
By decomposing $X$ and $\left|\Psi^{k}\right\rangle\left\langle\Psi^{k}\right|$, we find:
\begin{equation*}%\label{eq019}
\max_{\bar{s}_l}\operatorname{Tr}\left(\sum_j \mu_j G_j^A \otimes G_j^B \sum_l \tilde{s}_l H_l^A \otimes H_l^B\right) \leq \max_{\bar{s}} \sum_{j,l} \mu_j \tilde{s}_{l\downarrow}.
\end{equation*}
Here, $\tilde{s}_{l\downarrow}$ are the decreasingly ordered operator Schmidt coefficients of $\left|\Psi^{k}\right\rangle\left\langle\Psi^{k}\right|$, given by the products of vector Schmidt coefficients $s_i$ of $\left|\Psi^{k}\right\rangle = \sum_{i=1}^{k} s_i|\alpha_i \beta_i\rangle$.
Furthermore, the coefficients $s_{1},\cdots, s_{k}$ obey the constraints $s_{1}^2+\cdots+ s_{k}^2=1$ and $s_{1}\geq  \cdots \geq  s_{k}\geq 0$.

Following a procedure akin to that of $k=3,4$, the Operator Schmidt decomposition (OSD) for the projector $\left|\psi_k\right\rangle\left\langle\psi_k\right|$ is defined by the set
\begin{equation*}
  \left\{\tilde{\mu}_j\right\}=\left\{s_1 s_1, \cdots, s_1 s_k, s_2^2, \cdots, s_{k-1} s_k, s_k^2\right\},
\end{equation*}
For the function $f_{i}\left(s_1, \cdots, s_k\right)=\sum_{i j} \mu_i \tilde{\mu}_j$ to be derived, it is necessary to arrange the OSD coefficients $\tilde{\mu}_j$ in descending order.
The coefficient $\lambda_{k+1}$ can be obtained as the maximum of the following functions:
\begin{equation*}%\label{eqfk2}
  \max_{s_1, \cdots, s_k}\left\{f_i\left(s_1, \cdots, s_k\right)\right\},
\end{equation*}
where $i$ depends on the number of arrangements for $\{\tilde{\mu}_j\}$.

There are always such orderings for $\{s_{i} s_{j}\}_{1\le i,j\le k}$:
\begin{equation*}
  s_{1}^{2}\geq s_{1} s_{2}=s_{2} s_{1}\geq \cdots\geq  s_{1}s_{m} \geq  s_{2}^{2}\geq s_{2}s_{3}\geq  \cdots\geq s_{k}^{2},
\end{equation*}
there are corresponding $k$ order square matrix $M_u^{k}$ as following:
\begin{equation*}
M_u^{k}=
  \begin{pmatrix}
   \mu_1 & \mu_2 & \cdots & \mu_{2k-4}&\mu_{2k-2}\\
    \mu_3 & \mu_{2k} & \cdots & \mu_{4k-5} &\mu_{4k-3} \\
    \vdots & \vdots & \vdots & \vdots &\vdots\\
    \mu_{2k-3}& \mu_{4k-6} & \cdots & \mu_{k^2-3} & \mu_{k^2-2} &\\
    \mu_{2k-1} & \mu_{4k-4} & \cdots & \mu_{k^2-1} &\mu_{k^2}\\
  \end{pmatrix}.
\end{equation*}
It is not difficult to find that the first and k-th row of matrix $M_u^{k}$ exhibit the largest sum of the row values.

It is not difficult to find that all matrices $M_i^k$ are nonnegative matrices, based on \eqref{NMB4}, there are the following bounds:
\begin{equation*}
  \theta_{k+1}(M_i^k)= P_{k+1}(M_i^k)-\mu_{k^2}+\mu_{k^2}\frac{2(p_{k+1}(M_i^k)-\mu_{k^2})}{P_{k+1}(M_i^k)-2 \mu_{k^2}+\sqrt{(P_{k+1}(M_i^k))^2-4 \mu_{k^2}(P_{k+1}(M_i^k)-p_{k+1}(M_i^k))}}.
\end{equation*}

We define the binary function $g(P, p)$ as
\begin{equation*}
  g(P, p)=P-m+m\frac{2(p-m)}{P-2 m+\sqrt{P^2-4 m(P-p)}}
\end{equation*}
where $m=\mu_{k^2}$ denotes the smallest entry in the matrices, $P=P_{k+1}(M_i^k)$ and $p=p_{k+1}(M_i^k)$.  Furthermore, the partial derivatives are as follows:
\begin{small}
$$\frac{\partial g}{\partial P}= 1 - 2m(p - m)\frac{(P - 2m)}{\sqrt{P^2 - 4m(P - p)}} \frac{1}{\left(P - 2m + \sqrt{P^2 - 4m(P - p)}\right)^2}\geq 0,$$
$$\frac{\partial g}{\partial p}=  2m \left[(P - 2m + \sqrt{P^2 - 4m(P - p)}) - (p - m)\frac{2m}{\sqrt{P^2 - 4m(P - p)}}\right]  \frac{1}{\left(P - 2m + \sqrt{P^2 - 4m(P - p)}\right)^2}\geq 0,$$
\end{small}
which show that the binary function $g(P, p)$ exhibits increasing behavior concerning the variables $P$ and $p$.

It follows that
\begin{eqnarray*}
  \max\limits_{i}\{\theta_{k+1}(M_i^k)\} = \theta_{k+1}(M_u^{k})= \theta_{k+1}.
\end{eqnarray*}

Specifically, in the situation as following:
\begin{equation*}
  P=P_{k+1}(M_u^k)=\mu_1+\mu_2+\dots+\mu_{2k-2},\quad  p=p_{k+1}(M_u^k)=\mu_{2k-1}+\mu_{4k-4}+\cdots+\mu_{k^2-1}+\mu_{k^2},
\end{equation*}
the function $g(P, p)$ attains its maximum value.

Subsequently, we construct the following $(k+1)$-Schmidt witness:
\begin{equation*}
  \mathcal{SW}_{(k+1) }=\theta_{k+1} \mathbb{I}-X.
\end{equation*}
Consequently, the Theorem \ref{thmthetamn} can be drawn. The proof is completed.

\subsection{Proof of Proposition \ref{thmetamn} and Proposition \ref{thmzetamn}}

\subsubsection{Proof of Proposition \ref{thmetamn}}

Following a similar approach to the proof of Theorem \ref{thmthetamn}, we focus solely on the arrangement of OSD coefficients $\{\tilde{\mu}_j\}$ for $M_u^{k}$ that demonstrate the highest sum of values from the first and $k$-th rows of the matrix $M_u^{k}$.

It is not difficult to find that all matrices $M_i^k$ are nonnegative matrices, based on \eqref{NMB2}, there are the following bounds:
$$ \rho(M_i^k) \leq \eta_{k+1}(M_i^k)= P(M_i^k)-m(M_i^k)(1-\sqrt{\delta(M_i^k)}),$$
where $m(M_i^k)=\min\limits _{r, s}\left((M_i^k)_{r s}\right)$ and $\delta(M_i^k)=\max\limits _{P_r(M_i^k)<P_s(M_i^k)}\left(\frac{P_r(M_i^k)}{P_s(M_i^k)}\right)$.

Consequently, we have
\begin{equation*}
  \eta_{k+1}(M_u^{k})=  \mu_1+\mu_2+\dots+\mu_{2k-2}-\mu_{k^2}\left(1-\sqrt{\frac{\mu_{2k-1}+\mu_{4k-4}+\cdots+\mu_{k^2-1}+\mu_{k^2}}{\mu_1+\mu_2+\dots+\mu_{2k-2}}}\right).
\end{equation*}
%and
%\begin{equation*}
%  \eta_{k+1}(M_v^{k})=  \mu_1+\mu_2+\dots+\mu_{(k-2)^2+1}+\mu_{(k-1)^2+1}-\mu_{k^2}\left(1-\sqrt{\frac{\mu_{(k-1)^2+2} + \mu_{(k-1)^2+4} + \cdots + \mu_{k^2-1} +\mu_{k^2}}{\mu_1+\mu_2+\dots+\mu_{(k-1)^2+1}}}\right).
%\end{equation*}

%Following a similar approach as for $k=4,5$,
%We define the function
%$$g(x)=a+x-\mu_{k^2}+(\frac{b\mu_{k^2}^2}{a+x})^{\frac{1}{2}},\quad  x\in [\mu_{(k-1)^2+1}, \mu_{2k-2}],$$
%where $a+x=P_s(M_i^k)$ and $b=P_r(M_i^k)$. Consequently, there are
%$$g'(x)=1-(\frac{P_s(M_i^k)}{P_r(M_i^k)\mu_{k^2}^2})^{\frac{1}{2}}=1-\frac{1}{\mu_{k^2}}\sqrt{\frac{P_s(M_i^k)}{P_r(M_i^k)}}<0.$$

We define the function
$$g(P,p)=P-\mu_{k^2}(1-(\frac{p}{P})^{\frac{1}{2}}),$$
where $P=P_{k+1}(M_i^k)$ and $p=p_{k+1}(M_i^k)$. Consequently, there are
$$\frac{\partial g}{\partial P}=1-\frac{\mu_{k^2}}{2P}\sqrt{\frac{p}{P}}>0,\quad \frac{\partial g}{\partial p}=\frac{\mu_{k^2}}{2P}\frac{1}{\sqrt{Pp}}>0.$$

Because the function $g(x)$ is non-decreasing and the operator Schmidt coefficients $\mu_i$ are arranged in descending order, we only need to analyze the two extreme cases of the upper bounds of the spectral radii, namely, those of matrices $M_u^{k}$ and $M_v^{k}$.
It follows that
\begin{eqnarray*}
  \eta_{k+1}(M_u^{k}) & & =  \mu_1+\mu_2+\dots+\mu_{2k-2}-\mu_{k^2}\left(1-\sqrt{\frac{\mu_{2k-1}+\mu_{4k-4}+\cdots+\mu_{k^2-1}+\mu_{k^2}}{\mu_1+\mu_2+\dots+\mu_{2k-2}}}\right)\\
   %& & \leq \mu_1+\mu_2+\dots+\mu_{2k-2}-\mu_{k^2}\left(1-\sqrt{\frac{\mu_{2k-1}+\mu_{4k-4}+\cdots+\mu_{k^2-1}+\mu_{k^2}}{\mu_1+\mu_2+\dots+\mu_{2k-2}}}\right)\\
   &&= \eta_{k+1},
\end{eqnarray*}
%and
%\begin{eqnarray*}
%  \eta_{k+1}(M_v^{k}) & &=  \mu_1+\mu_2+\dots+\mu_{2k-2}-\mu_{k^2}\left(1-\sqrt{\frac{\mu_{(k-1)^2+2} + \mu_{(k-1)^2+4} + \cdots  +\mu_{k^2}}{\mu_1+\mu_2+\dots+\mu_{(k-1)^2+1}}}\right)\\
%   & & \leq \mu_1+\mu_2+\dots+\mu_{2k-2}-\mu_{k^2}\left(1-\sqrt{\frac{\mu_{2k-1}+\mu_{4k-4}+\cdots+\mu_{k^2-1}+\mu_{k^2}}{\mu_1+\mu_2+\dots+\mu_{(k-1)^2+1}}}\right)\\
%   &&=\eta_{k+1}.
%\end{eqnarray*}

It follows that
\begin{eqnarray*}
  \max\limits_{i}\eta_{k+1}(M_i^k) = \max\{\eta_{k+1}(M_u^{k})\}= \eta_{k+1}.
\end{eqnarray*}

In other words, the function $g(P, p)$ attains its maximum value in case of
\begin{equation*}
  P=P_{k+1}(M_u^k),\quad  p=p_{k+1}(M_u^k),
\end{equation*}

Subsequently, we construct the following $(k+1)$-Schmidt witness:
\begin{equation*}
  \mathcal{SW}_{(k+1) }=\eta_{k+1} \mathbb{I}-X.
\end{equation*}
Consequently, the Proposition \ref{thmetamn} can be drawn. The proof is completed.

%Therefore, there are $\lambda_{k+1}\leq \eta_{k+1} \leq  \mu_1+\mu_2+\dots+\mu_{2k-2}$ and we construct the $(k+1)$-Schmidt witness $\mathcal{SW}_{(k+1)}=\eta_{k+1} \mathbb{I}-X$.
% Thus, the proof is completed. %$\Box$

\subsubsection{Proof of Proposition \ref{thmzetamn}}

Following a similar approach to the proof of Proposition \ref{thmetamn}, invoking inequality \eqref{NMB3} allows us to expeditiously present the ensuing conclusion:
\begin{equation*}
  \max\limits_{i}\zeta_{k+1}(M_i^k) = \max\{\zeta_{k+1}(M_u^{k})\}= \zeta_{k+1}.
\end{equation*}

In other words, the function $g(P, p)$ attains its maximum value in case of
\begin{equation*}
  P=P_{k+1}(M_u^k),\quad  p=p_{k+1}(M_u^k),
\end{equation*}

Subsequently, we construct the following $(k+1)$-Schmidt witness:
\begin{equation*}
  \mathcal{SW}_{(k+1) }=\zeta_{k+1} \mathbb{I}-X.
\end{equation*}
Consequently, the Proposition \ref{thmzetamn} can be drawn. The proof is completed.

It is noting that
\begin{equation*}
  \frac{\mu_{2k-1}+\mu_{4k-4}+\cdots+\mu_{k^2-1}}{\mu_1+\mu_2+\dots+\mu_{(k-1)^2+1}-\mu_{k^2}}\leq \frac{\mu_{2k-1}+\mu_{4k-4}+\cdots+\mu_{k^2-1}+\mu_{k^2}}{\mu_1+\mu_2+\dots+\mu_{(k-1)^2+1}},
\end{equation*}
which implies that the bound for $\zeta_{k+1}$  is tighter than bound for $\eta_{k+1}$.

\subsection{Characterizing Schmidt number witnesses for arbitrary dimensional bipartite quantum states in Schmidt number four}

In general, the precise mathematical expression for the coefficient $\lambda_4$ in the presence of the operator Schmidt coefficients $\mu_1, \ldots, \mu_{9}$ eludes us.
It is worth noting that the matrix set $\{M_i\}$ comprises solely two matrices for $k=4$.
Consequently, we contemplate employing the nonnegative matrix spectral radius theory to establish  upper bounds for the coefficient $\lambda_4$ in Remark \ref{thmetamn}, \ref{thmzetamn} and Theorem \ref{thmthetamn}, followed by constructing the corresponding Schmidt number witness $\mathcal{SW}_{4}$.

\begin{corollary}\label{corSN41}
Consider a hermitian operator $X$ in $\mathcal{L}(\mathcal{H}_m) \otimes \mathcal{L}(\mathcal{H}_n)$  $(n\geq m\geq 4)$,  the optimal $4$-Schmidt witness based on the OSD is $\mathcal{SW}_{4 }=\lambda_{4} \mathbb{I}-X$, with $\lambda_{4} = \max _{\left|\Psi^{3}\right\rangle \in S_{3}}\left\langle\Psi^{3}|X| \Psi^{3}\right\rangle$ and $S_{3}$ denotes the states set with Schmidt number $3$. We define the $\theta_{4}$ as following:
\begin{equation*}
\theta_4 =\mu_1+\mu_2+\mu_4+\mu_9-\left(\frac{2\mu_9(\mu_5+\mu_8)}{\mu_1+\mu_2+\mu_4-2\mu_9+\sqrt{(\mu_1+\mu_2+\mu_4)^2-4\mu_9(\mu_1+\mu_2+\mu_4-\mu_5-\mu_8-\mu_9)}}\right),
\end{equation*}
where $\mu_i, i=1,\cdots,9$ are Operator Schmidt coefficients of $X$,  $P_{4}=\mu_1+\mu_2+\mu_4$ and $p_{4}=\mu_{5}+\mu_{8}+\mu_{9}$,
Notably,
\begin{equation*}
  \lambda_4\leq \theta_4.
\end{equation*}
Subsequently, we construct the following $4$-Schmidt witness:
\begin{equation*}
  \mathcal{SW}_{4}=\theta_4 \mathbb{I}-X.
\end{equation*}
\end{corollary}

\begin{corollary}
Using the same notations as in Corollary \ref{corSN41}. We define the $\eta_{4}$ as following:
\begin{equation*}
  \eta_4 =\mu_1+\mu_2+\mu_4-\mu_9\left(1-\sqrt{\frac{\mu_5+\mu_8+\mu_9}{\mu_1+\mu_2+\mu_4}}\right).
\end{equation*}
Notably, we have the following inequality:
\begin{equation*}
\lambda_4\leq \eta_4 .
\end{equation*}
Subsequently, we construct the following $4$-Schmidt witness:
\begin{equation*}
  \mathcal{SW}_{4}=\eta_4 \mathbb{I}-X.
\end{equation*}
\end{corollary}

%In fact, the literature \cite{PRL24OSD} demonstrates that the coefficient $\lambda_4$ can be obtained as the maximum of the maximum eigenvalue of the symmetrized versions of matrices $M_1^4$ and $M_2^4$ as  \eqref{matricesM42}.
%Based on inequalities \eqref{NMB2}, there is
%\begin{equation*}
%   \lambda_4\leq \max \left\{\mu_1+\mu_2+\mu_4-\mu_9\left(1-\sqrt{\frac{\mu_5+\mu_8+\mu_9}{\mu_1+\mu_2+\mu_4}}\right), \mu_1+\mu_2+\mu_5-\mu_9\left(1-\sqrt{\frac{\mu_5+\mu_8+\mu_9}{\mu_1+\mu_2+\mu_5}}\right)\right\}.
%\end{equation*}
%
%We construct the function $$g(x)=a+x-c+(\frac{bc^2}{a+x})^{\frac{1}{2}}, x\in [\mu_5, \mu_4],$$ where $a=\mu_1+\mu_2$, $b=\mu_5+\mu_8+\mu_9$ and $c=\mu_9$. Consequently, there is $$g'(x)=1-(\frac{a+x}{bc^2})^{\frac{1}{2}}=1-\frac{1}{c}\sqrt{\frac{a+x}{b}}.$$
%
%According to Lemma \ref{lem3.1mu} and the fact $a+x\geq b$, there is $g'(x)\leq 0.$
%It yields
%\begin{equation*}
%  \lambda_4\leq \mu_1+\mu_2+\mu_5-\mu_9\left(1-\sqrt{\frac{\mu_5+\mu_8+\mu_9}{\mu_1+\mu_2+\mu_5}}\right)\leq \mu_1+\mu_2+\mu_5\leq \mu_1+\mu_2+\mu_4.
%\end{equation*}
%
%Furthermore, there is the $4$-Schmidt witness $\mathcal{W}_{4 \mathrm{S}}=\eta_4 \mathbb{I}-X.$

\begin{corollary}
Using the same notations as in Corollary \ref{corSN41}. We define the $\eta_{4}$ as following:
\begin{equation*}
 \zeta_4 =\mu_1+\mu_2+\mu_4-\mu_9\left(1-\sqrt{\frac{\mu_5+\mu_8}{\mu_1+\mu_2+\mu_4-\mu_9}}\right).
\end{equation*}
Notably, within this context, we have the following inequality:
\begin{equation*}
  \lambda_4\leq \zeta_4.
\end{equation*}
Subsequently, we construct the following $4$-Schmidt witness:
\begin{equation*}
  \mathcal{SW}_{4}=\zeta_4 \mathbb{I}-X.
\end{equation*}
\end{corollary}

From Example 1, if $k=3$, there is $$\rho_3=\frac{1}{2}|\phi^+\rangle \langle \phi^+|+\frac{1}{4}(|12\rangle+|21\rangle)(\langle 12|+\langle 21|),$$ with $|\phi^+\rangle=\frac{1}{\sqrt{3}}(|00\rangle+|11\rangle+|22\rangle).$

There are following the operator Schmidt coefficients(OSC):
$$\mu_1=\mu_2=0.4167,  \mu_3=\mu_4=\mu_5=\mu_6=\mu_7=0.1667,  \mu_8=\mu_9=0.0833.$$

Therefore, there are
$$\lambda_{4}=0.6642, \theta_{4}=0.9420, \zeta_{4}=0.9603 , \zeta_{4}=0.9649, P_{4}=1.0001.$$
%\begin{remark}
%It is noting that $\theta_4\leq \zeta_4\leq \eta_4$, the $4$-Schmidt witness $\mathcal{W}_{4 \mathrm{S}}=\zeta_4 \mathbb{I}-X$ is a more effective $4$-Schmidt witness than $\mathcal{W}_{4 \mathrm{S}}=\eta_4 \mathbb{I}-X$, and identify a broader range of entangled states.
%\end{remark}

\subsection{Characterizing Schmidt number witnesses for arbitrary dimensional bipartite quantum states in Schmidt number five }

Let us now focus on the identification of witnesses certifying the Schmidt number $k=5$. Our objective is to determine the coefficient $\lambda_5$ such that the Schmidt witness yields a nonnegative expectation value for all states $\left|\psi_4\right\rangle$ with $\mathrm{SN} (k-1)=4$.

Following a procedure similar to that of $k=4$, the Operator Schmidt Decomposition (OSD) of the projector $\left|\psi_4\right\rangle\left\langle\psi_4\right|$ is defined by the set
\begin{equation*}
  \left\{\tilde{\mu}_j\right\}=\left\{s_1 s_1, s_1 s_2, s_1 s_3, s_1 s_4, s_2 s_2, s_2 s_3, s_2 s_4,s_3 s_3, s_3 s_4\right\},
\end{equation*}
In order to derive the function $f_{i}\left(s_1, s_2, s_3, s_4\right)=\sum_{i j} \mu_i \tilde{\mu}_j$, the OSD coefficients $\tilde{\mu}_j$ must be arranged in descending order.
Since $s_1\geq  s_2\geq  s_3 \geq  s_4\geq 0$, there are twelve possible orderings for the numbers $\{s_{i} s_{j}\}_{1\le i,j\le 4}$:

$$s_{1}^{2}\geq s_{1} s_{2}\geq s_{1} s_{3}\geq  s_{1}s_{4} \geq  s_{2}^{2}\geq s_{2}s_{3}\geq  \max{(s_{2}s_{4}, s_{3}^{2})}\geq \min{(s_{2}s_{4}, s_{3}^{2})}\geq s_{3}s_{4}\geq s_{4}^{2},$$
$$s_{1}^{2}\geq s_{1} s_{2}\geq \max{(s_{1} s_{3}, s_{2}^{2})}\geq \min{(s_{1} s_{3}, s_{2}^{2})}\geq s_{2}s_{3}\geq s_{3}^{2}\geq s_{1}s_{4}\geq s_{2}s_{4}\geq s_{3}s_{4}\geq s_{4}^{2},$$
$$s_{1}^{2}\geq s_{1} s_{2}\geq \max{(s_{1} s_{3}, s_{2}^{2})}\geq \min{(s_{1} s_{3}, s_{2}^{2})}\geq s_{2}s_{3}\geq s_{1}s_{4}\geq s_{3}^{2}\geq s_{2}s_{4}\geq s_{3}s_{4}\geq s_{4}^{2},$$
$$s_{1}^{2}\geq s_{1} s_{2}\geq \max{(s_{1} s_{3}, s_{2}^{2})}\geq \min{(s_{1} s_{3}, s_{2}^{2})}\geq s_{2}s_{3}\geq s_{1}s_{4}\geq s_{2}s_{4}\geq s_{3}^{2}\geq s_{3}s_{4}\geq s_{4}^{2},$$
$$s_{1}^{2}\geq s_{1} s_{2}\geq \max{(s_{1} s_{3}, s_{2}^{2})}\geq \min{(s_{1} s_{3}, s_{2}^{2})}\geq s_{1}s_{4}\geq s_{2}s_{3}\geq
s_{3}^{2}\geq s_{2}s_{4}\geq s_{3}s_{4}\geq s_{4}^{2},$$
$$s_{1}^{2}\geq s_{1} s_{2}\geq \max{(s_{1} s_{3}, s_{2}^{2})}\geq \min{(s_{1} s_{3}, s_{2}^{2})}\geq s_{1}s_{4}\geq s_{2}s_{3}\geq  s_{2}s_{4}\geq s_{3}^{2}\geq s_{3}s_{4}\geq s_{4}^{2}.$$

This leads to twelve possible functions $f_{i}\left(s_1, s_2, s_3, s_4\right), i=1\cdots 12$ that have to be considered:
$$\begin{aligned}
 f_{1}\left(s_1, s_2, s_3,s_4\right)=& \mu_1 s_1^2+(\mu_2+\mu_3) s_1 s_2+(\mu_4+\mu_5) s_1 s_3+(\mu_6+\mu_7) s_1 s_4 +\mu_8 s_2^2+(\mu_9+\mu_{10}) s_2 s_3 \\
& +(\mu_{11}+\mu_{12})s_2 s_4+ \mu_{13}s_3^2+(\mu_{14}+\mu_{15})s_3 s_4+ \mu_{16}s_4^2,
\end{aligned}$$
$$\begin{aligned}
 f_{2}\left(s_1, s_2, s_3,s_4\right)=& \mu_1 s_1^2+(\mu_2+\mu_3) s_1 s_2+(\mu_4+\mu_5) s_1 s_3+(\mu_6+\mu_7) s_1 s_4 +\mu_8 s_2^2+(\mu_9+\mu_{10}) s_2 s_3 \\
& +(\mu_{12}+\mu_{13})s_2 s_4+ \mu_{11}s_3^2+(\mu_{14}+\mu_{15})s_3 s_4+ \mu_{16}s_4^2,
\end{aligned}$$
$$\begin{aligned}
 f_{3}\left(s_1, s_2, s_3,s_4\right)=& \mu_1 s_1^2+(\mu_2+\mu_3) s_1 s_2+(\mu_4+\mu_5) s_1 s_3+(\mu_{10}+\mu_{11}) s_1 s_4 +\mu_6 s_2^2+(\mu_7+\mu_{8}) s_2 s_3 \\
& +(\mu_{12}+\mu_{13})s_2 s_4+ \mu_{9}s_3^2+(\mu_{14}+\mu_{15})s_3 s_4+ \mu_{16}s_4^2,
\end{aligned}$$
$$\begin{aligned}
 f_{4}\left(s_1, s_2, s_3,s_4\right)=& \mu_1 s_1^2+(\mu_2+\mu_3) s_1 s_2+(\mu_5+\mu_6) s_1 s_3+(\mu_{10}+\mu_{11}) s_1 s_4 +\mu_4 s_2^2+(\mu_7+\mu_{8}) s_2 s_3 \\
& +(\mu_{12}+\mu_{13})s_2 s_4+ \mu_{9}s_3^2+(\mu_{14}+\mu_{15})s_3 s_4+ \mu_{16}s_4^2,
\end{aligned}$$
$$\begin{aligned}
 f_{5}\left(s_1, s_2, s_3,s_4\right)=& \mu_1 s_1^2+(\mu_2+\mu_3) s_1 s_2+(\mu_4+\mu_5) s_1 s_3+(\mu_9+\mu_{10}) s_1 s_4 +\mu_6 s_2^2+(\mu_7+\mu_{8}) s_2 s_3 \\
& +(\mu_{12}+\mu_{13})s_2 s_4+ \mu_{11}s_3^2+(\mu_{14}+\mu_{15})s_3 s_4+ \mu_{16}s_4^2,
\end{aligned}$$
$$\begin{aligned}
 f_{6}\left(s_1, s_2, s_3,s_4\right)=& \mu_1 s_1^2+(\mu_2+\mu_3) s_1 s_2+(\mu_5+\mu_6) s_1 s_3+(\mu_9+\mu_{10}) s_1 s_4 +\mu_4 s_2^2+(\mu_7+\mu_{8}) s_2 s_3 \\
& +(\mu_{12}+\mu_{13})s_2 s_4+ \mu_{11}s_3^2+(\mu_{14}+\mu_{15})s_3 s_4+ \mu_{16}s_4^2,
\end{aligned}$$
$$\begin{aligned}
 f_{7}\left(s_1, s_2, s_3,s_4\right)=& \mu_1 s_1^2+(\mu_2+\mu_3) s_1 s_2+(\mu_4+\mu_5) s_1 s_3+(\mu_9+\mu_{10}) s_1 s_4 +\mu_6 s_2^2+(\mu_7+\mu_{8}) s_2 s_3 \\
& +(\mu_{11}+\mu_{12})s_2 s_4+ \mu_{13}s_3^2+(\mu_{14}+\mu_{15})s_3 s_4+ \mu_{16}s_4^2,
\end{aligned}$$
$$\begin{aligned}
 f_{8}\left(s_1, s_2, s_3,s_4\right)=& \mu_1 s_1^2+(\mu_2+\mu_3) s_1 s_2+(\mu_5+\mu_6) s_1 s_3+(\mu_9+\mu_{10}) s_1 s_4 +\mu_4 s_2^2+(\mu_7+\mu_{8}) s_2 s_3 \\
& +(\mu_{11}+\mu_{12})s_2 s_4+ \mu_{13}s_3^2+(\mu_{14}+\mu_{15})s_3 s_4+ \mu_{16}s_4^2,
\end{aligned}$$
$$\begin{aligned}
 f_{9}\left(s_1, s_2, s_3,s_4\right)=& \mu_1 s_1^2+(\mu_2+\mu_3) s_1 s_2+(\mu_4+\mu_5) s_1 s_3+(\mu_7+\mu_{8}) s_1 s_4 +\mu_6 s_2^2+(\mu_9+\mu_{10}) s_2 s_3 \\
& +(\mu_{12}+\mu_{13})s_2 s_4+ \mu_{11}s_3^2+(\mu_{14}+\mu_{15})s_3 s_4+ \mu_{16}s_4^2,
\end{aligned}$$
$$\begin{aligned}
 f_{10}\left(s_1, s_2, s_3,s_4\right)=& \mu_1 s_1^2+(\mu_2+\mu_3) s_1 s_2+(\mu_5+\mu_6) s_1 s_3+(\mu_7+\mu_{8}) s_1 s_4 +\mu_4 s_2^2+(\mu_9+\mu_{10}) s_2 s_3 \\
& +(\mu_{12}+\mu_{13})s_2 s_4+ \mu_{11}s_3^2+(\mu_{14}+\mu_{15})s_3 s_4+ \mu_{16}s_4^2,
\end{aligned}$$
$$\begin{aligned}
 f_{11}\left(s_1, s_2, s_3,s_4\right)=& \mu_1 s_1^2+(\mu_2+\mu_3) s_1 s_2+(\mu_4+\mu_5) s_1 s_3+(\mu_7+\mu_{8}) s_1 s_4 +\mu_6 s_2^2+(\mu_9+\mu_{10}) s_2 s_3 \\
& +(\mu_{11}+\mu_{12})s_2 s_4+ \mu_{13}s_3^2+(\mu_{14}+\mu_{15})s_3 s_4+ \mu_{16}s_4^2,
\end{aligned}$$
$$\begin{aligned}
 f_{12}\left(s_1, s_2, s_3,s_4\right)=& \mu_1 s_1^2+(\mu_2+\mu_3) s_1 s_2+(\mu_5+\mu_6) s_1 s_3+(\mu_7+\mu_{8}) s_1 s_4 +\mu_4 s_2^2+(\mu_9+\mu_{10}) s_2 s_3 \\
& +(\mu_{11}+\mu_{12})s_2 s_4+ \mu_{13}s_3^2+(\mu_{14}+\mu_{15})s_3 s_4+ \mu_{16}s_4^2.
\end{aligned}$$

The coefficient $\lambda_5$ is subsequently determined as the maximum among these twelve functions:
\begin{equation}\label{LOfuction}
  \max\limits_{s_1, s_2, s_3, s_4}\{f_{i}\left(s_1, s_2, s_3, s_4\right), i=1\cdots 12\}.
\end{equation}

This can be achieved by expressing $f_{i}(s_1, s_2, s_3, s_4)$ in matrices terms, following a methodology akin to that employed for $k =3,4$. The corresponding matrices are provided by
\begin{equation*}
M_1^4=
  \begin{pmatrix}
    \mu_1 & \mu_2 & \mu_4 & \mu_6\\
    \mu_3 & \mu_8 & \mu_9 & \mu_{11} \\
    \mu_5 & \mu_{10} & \mu_{13} & \mu_{14}\\
    \mu_7 & \mu_{12} & \mu_{15} & \mu_{16}\\
  \end{pmatrix}, \quad
M_2^4=
  \begin{pmatrix}
     \mu_1 & \mu_2 & \mu_4 & \mu_6\\
    \mu_3 & \mu_8 & \mu_9 & \mu_{12} \\
    \mu_5 & \mu_{10} & \mu_{11} & \mu_{14}\\
    \mu_7 & \mu_{13} & \mu_{15} & \mu_{16}\\
  \end{pmatrix},\quad
M_3^4=
  \begin{pmatrix}
    \mu_1 & \mu_2 & \mu_4 & \mu_{10}\\
    \mu_3 & \mu_6 & \mu_7 & \mu_{12} \\
    \mu_5 & \mu_{8} & \mu_{9} & \mu_{14}\\
    \mu_{11} & \mu_{13} & \mu_{15} & \mu_{16}\\
  \end{pmatrix},
\end{equation*}
\begin{equation*}
M_4^4=
  \begin{pmatrix}
   \mu_1 & \mu_2 & \mu_5 & \mu_{10}\\
    \mu_3 & \mu_4 & \mu_7 & \mu_{12} \\
    \mu_6 & \mu_{8} & \mu_{9} & \mu_{14}\\
    \mu_{11} & \mu_{13} & \mu_{15} & \mu_{16}\\
  \end{pmatrix}, \quad
M_5^4=
  \begin{pmatrix}
   \mu_1 & \mu_2 & \mu_4 & \mu_{9}\\
    \mu_3 & \mu_6 & \mu_7 & \mu_{12} \\
    \mu_5 & \mu_{8} & \mu_{11} & \mu_{14}\\
    \mu_{10} & \mu_{13} & \mu_{15} & \mu_{16}\\
  \end{pmatrix}, \quad
M_6^4=
  \begin{pmatrix}
    \mu_1 & \mu_2 & \mu_5 & \mu_{9}\\
    \mu_3 & \mu_4 & \mu_7 & \mu_{12} \\
    \mu_6 & \mu_{8} & \mu_{11} & \mu_{14}\\
    \mu_{10} & \mu_{13} & \mu_{15} & \mu_{16}\\
  \end{pmatrix},
\end{equation*}
\begin{equation*}
M_7^4=
  \begin{pmatrix}
    \mu_1 & \mu_2 & \mu_4 & \mu_{9}\\
    \mu_3 & \mu_6 & \mu_7 & \mu_{11} \\
    \mu_5 & \mu_{8} & \mu_{13} & \mu_{14}\\
    \mu_{10} & \mu_{12} & \mu_{15} & \mu_{16}\\
  \end{pmatrix}, \quad
M_8^4=
  \begin{pmatrix}
    \mu_1 & \mu_2 & \mu_5 & \mu_{9}\\
    \mu_3 & \mu_4 & \mu_7 & \mu_{11} \\
    \mu_6 & \mu_{8} & \mu_{13} & \mu_{14}\\
    \mu_{10} & \mu_{12} & \mu_{15} & \mu_{16}\\
  \end{pmatrix},\quad
M_9^4=
  \begin{pmatrix}
    \mu_1 & \mu_2 & \mu_4 & \mu_{7}\\
    \mu_3 & \mu_6 & \mu_9 & \mu_{12} \\
    \mu_5 & \mu_{10} & \mu_{11} & \mu_{14}\\
    \mu_{8} & \mu_{13} & \mu_{15} & \mu_{16}\\
  \end{pmatrix},
\end{equation*}
\begin{equation*}
M_{10}^4=
  \begin{pmatrix}
    \mu_1 & \mu_2 & \mu_5 & \mu_{7}\\
    \mu_3 & \mu_4 & \mu_9 & \mu_{12} \\
    \mu_6 & \mu_{10} & \mu_{11} & \mu_{14}\\
    \mu_{8} & \mu_{13} & \mu_{15} & \mu_{16}\\
  \end{pmatrix},\quad
M_{11}^4=
  \begin{pmatrix}
    \mu_1 & \mu_2 & \mu_4 & \mu_{7}\\
    \mu_3 & \mu_6 & \mu_9 & \mu_{11} \\
    \mu_5 & \mu_{10} & \mu_{13} & \mu_{14}\\
    \mu_{8} & \mu_{12} & \mu_{15} & \mu_{16}\\
  \end{pmatrix}, \quad
M_{12}^4=
  \begin{pmatrix}
    \mu_1 & \mu_2 & \mu_5 & \mu_{7}\\
    \mu_3 & \mu_4 & \mu_9 & \mu_{11} \\
    \mu_6 & \mu_{10} & \mu_{13} & \mu_{14}\\
    \mu_{8} & \mu_{12} & \mu_{15} & \mu_{16}\\
  \end{pmatrix}.
\end{equation*}
Following the established methodology, the coefficient $\lambda_5$ can be expressed as
$$
\lambda_5=\max \left\{\operatorname{maxEig}\left(M_i^{\text {symm }}\right), i=1, \ldots, 12\right\},
$$
where $\operatorname{maxEig}\left(M_i^{\text {symm }}\right)$ denotes the maximum eigenvalue of $M_i^{\text {symm }}$, and $M_i^{\text {symm }}$ represent the symmetrized forms of matrix $M_i$ as defined in Equation \eqref{LOfuction}. Since the $M_i$ matrices are $4 \times 4$ matrices, $\lambda_5$ is the solution to a fourth-order polynomial in the variables $\mu_i(i=1, \ldots, 16)$, which correspond to the initial sixteen Operator Schmidt Coefficients (OSC) of the operator $X$.

It is not difficult to find that all matrices $M_i$ are nonnegative matrices, based on Theorem \ref{thmthetamn} and \eqref{NMB2},\eqref{NMB3}, \eqref{NMB4}, there are the following bounds.

\begin{corollary}\label{corSN51}
Consider a hermitian operator $X$ in $\mathcal{L}(\mathcal{H}_m) \otimes \mathcal{L}(\mathcal{H}_n)$ $(n\geq m\geq 5)$,  the optimal $5$-Schmidt witness based on the OSD is $\mathcal{SW}_{5 }=\lambda_{5} \mathbb{I}-X$, with $\lambda_{5} = \max _{\left|\Psi^{4}\right\rangle \in S_{4}}\left\langle\Psi^{4}|X| \Psi^{4}\right\rangle$ and $S_{4}$ denotes the states set with Schmidt number $4$. We define the $\theta_{5}$ as following:
\begin{eqnarray}\label{eqbound03}
  \theta_5 &=&  \mu_1+\mu_2+\mu_4+\mu_6-\mu_{16}\\
   &+& \left(\frac{2\mu_{16}(\mu_7+\mu_{12}+\mu_{15})}{\mu_1+\mu_2+\mu_4+\mu_{6}-2\mu_{16}+\sqrt{(\mu_1+\mu_2+\mu_4+\mu_{6})^2-4\mu_{16}(\mu_1+\mu_2+\mu_4+\mu_{6}-\mu_7-\mu_{12}-\mu_{15}-\mu_{16})}}\right) \nonumber.
\end{eqnarray}
where $\mu_i, i=1,\cdots,16$ are Operator Schmidt coefficients of $X$,  $P_{5}=\mu_1+\mu_2+\mu_4+\mu_{6}$ and $p_{5}=\mu_{7}+\mu_{12}+\mu_{15}+\mu_{16}$,
Notably,
\begin{equation*}
  \lambda_5\leq \theta_5.
\end{equation*}
Subsequently, we construct the following $5$-Schmidt witness:
\begin{equation*}
  \mathcal{SW}_{5}=\theta_5 \mathbb{I}-X.
\end{equation*}
\end{corollary}

\begin{corollary}
Using the same notations as in Corollary \ref{corSN51}.  We define the $\eta_{5}$ as following:
\begin{equation}\label{eqbound01}
 \eta_5 =\mu_1+\mu_2+\mu_4+\mu_{6}-\mu_{16}\left(1-\sqrt{\frac{\mu_7+\mu_{12}+\mu_{15}+\mu_{16}}{\mu_1+\mu_2+\mu_4+\mu_{6}}}\right).
\end{equation}
Notably, we have the following inequality:
\begin{equation*}
  \lambda_5\leq \eta_5.
\end{equation*}
%where $$\eta_5=\mu_1+\mu_2+\mu_5+\mu_{10}-\mu_{16}\left(1-\sqrt{\frac{\mu_7+\mu_{12}+\mu_{15}+\mu_{16}}{\mu_1+\mu_2+\mu_5+\mu_{10}}}\right).$$
Subsequently, we construct the following $5$-Schmidt witness:
\begin{equation*}
  \mathcal{SW}_{5}=\eta_5 \mathbb{I}-X.
\end{equation*}
\end{corollary}

%$Similarly, with reference to inequality \eqref{NMB3}, we can promptly articulate the subsequent conclusion.

\begin{corollary}
Using the same notations as in Corollary \ref{corSN51}.  We define the $\zeta_{5}$ as following:
\begin{equation}\label{eqbound02}
 \zeta_4 =\mu_1+\mu_2+\mu_4+\mu_{6}-\mu_{16}\left(1-\sqrt{\frac{\mu_7+\mu_{12}+\mu_{15}}{\mu_1+\mu_2+\mu_4+\mu_{6}-\mu_{16}}}\right).
\end{equation}
Notably, within this context, we have the following inequality:
\begin{equation*}
  \lambda_5\leq \zeta_4.
\end{equation*}
Subsequently, we construct the following $4$-Schmidt witness:
\begin{equation*}
  \mathcal{SW}_{5}=\zeta_5 \mathbb{I}-X.
\end{equation*}
\end{corollary}

From Example 1, if $k=4$, there is $$\rho_4=\frac{1}{2}|\phi^+\rangle \langle \phi^+|+\frac{1}{4}(|23\rangle+|32\rangle)(\langle 23|+\langle 32|),$$ with $|\phi^+\rangle=\frac{1}{\sqrt{3}}(|00\rangle+|11\rangle+|22\rangle+|33\rangle).$

There are following the operator Schmidt coefficients(OSC):
$$\mu_1=\mu_2=0.3750,   \mu_3=\cdots=\mu_{16}=0.1250.$$

Therefore, there are
$$\lambda_{5}=0.6545, \theta_{5}=0.9330, \zeta_{5}=0.9568 , \zeta_{5}=0.9634, P_{5}=1.0000.$$

%\begin{remark}
%When $s_1=s_2=s_3=s_4=\frac{1}{2}$, ie., $\sum\limits_{i=1}^{4}s_i^2=1$, there is
%$$f_{j}\left(s_1, s_2, s_3,s_4\right)=\frac{1}{4}\sum\limits_{i=1}^{16}\mu_i=\frac{1}{4}\sum\limits_{i=1}^{16}\mu_i\leq \mu_1 + \mu_2 + \mu_3 + \mu_4, \quad j=1,\dots,12.$$
%\end{remark}

\end{widetext}

\end{document}